%% file: ieeedebruit_v21_last.tex
\documentclass[11pt,draftcls,onecolumn]{IEEEtran}
\usepackage[T1]{fontenc}
\usepackage{ae,aecompl}
\usepackage{pslatex} %
\usepackage{graphicx}
\usepackage{amsmath}
\usepackage{amssymb}
\usepackage{cite}
\usepackage{bm}

\usepackage{dsfont}
\usepackage{color}

\newcommand{\E}{\mathsf{E}}

\newcommand{\NN}{{\ensuremath{\mathbb N}}}
\newcommand{\RR}{{\ensuremath{\mathbb R}}}
\newcommand{\ZZ}{{\ensuremath{\mathbb Z}}}
\newcommand{\KK}{\ensuremath{\mathbb K}}

\newcommand{\HH}{{\ensuremath{\mathrm H}}}
\newcommand{\kk}{{\ensuremath{\mathbf k}}}
\newcommand{\mm}{{\ensuremath{\mathbf m}}}

\newcommand{\GGamma}{{\ensuremath{\boldsymbol{\Gamma}}}}
\newcommand{\NM}[1][M]{\ensuremath{{\mathbb N}_{#1}}}
\newcommand{\NMSQ}[1][M]{\ensuremath{{\mathbb N}^{2}_{#1}}}
\newcommand{\NMs}[1][M]{\ensuremath{{\mathbb N}^{\star}_{#1}}}
\newcommand{\NMsSQ}[1][M]{\ensuremath{{\mathbb N}^{\star 2}_{#1}}}

\newcommand{\LLa}{\null}

\newcommand{\aaa}{{\ensuremath{\mathbf a}}}

\newcommand{\eqdef}{\mbox{$\buildrel\triangle\over =$}}

\newcommand{\est}[2]{\overset{\curlywedge}{#1}^{#2}\!\!\!}
\newcommand{\estw}[2]{\overset{\curlywedge}{#1}^{#2}\!\!\!\!\!\!\!}
\newcommand{\estwl}[2]{\overset{\curlywedge}{#1}^{#2}\!\!\!\!\!\!\!\!\!\!\!}
\newcommand{\ests}[1]{\overset{\curlywedge}{#1}}

\newcommand{\bbf}[1]{\overline{\mathbf #1}}

\newcommand{\vect}[1]{{\bf #1}}

\newtheorem{prop}{Proposition}

\begin{document}

\title{A Nonlinear Stein Based Estimator for Multichannel Image Denoising}

\author{Caroline Chaux, {\em Member IEEE}, Laurent Duval, {\em Member  IEEE,}\\ Amel Benazza-Benyahia, {\em Member IEEE,}  and Jean-Christophe Pesquet, {\em Senior Member  IEEE}
\thanks{C. Chaux and J.-C. Pesquet are with the Institut Gaspard Monge and CNRS-UMR 8049,
Universit{\'e} de Paris-Est, 77454 Marne-la-Vall{\'e}e
Cedex 2, France. E-mail: \texttt{\{chaux,pesquet\}@univ-mlv.fr}, phone: +33
1 60 95 77 39.}
\thanks{L. Duval is with the Institut Fran{\c c}ais du P{\'e}trole (IFP),
Technology, Computer Science and Applied Mathematics Division,
92500 Rueil-Malmaison, France. E-mail:
\texttt{laurent.duval@ifp.fr}, phone: +33 1 47 52 61 02.}
\thanks{A. Benazza-Benyahia is with the
Unit\'e de Recherche en Imagerie Satellitaire et ses Applications
(URISA), {\'E}cole Sup{\'e}rieure des Communications de Tunis,
Route de Raoued 3.5 Km, 2083 Ariana, Tunisia. E-mail: \texttt{
benazza.amel@supcom.rnu.tn}, phone: +216 1755621.}}

\maketitle

\begin{abstract}
The use of multicomponent images has become widespread with the
improvement of multisensor systems having increased spatial and
spectral resolutions. However, the observed images are  often
corrupted by an additive Gaussian  noise. In this paper, we are
interested in multichannel image denoising based on a multiscale
representation of the images. A multivariate statistical approach
is adopted to take into account both the spatial and the
inter-component correlations existing between the different
wavelet subbands. More precisely, we propose a new parametric
nonlinear estimator which generalizes many reported denoising
methods.  The derivation of the optimal parameters is  achieved by
applying Stein's principle in the multivariate case. Experiments
performed on multispectral remote sensing images clearly indicate
that our method outperforms conventional wavelet denoising
techniques.

\end{abstract}


\section{Introduction}
Many real world images  are contaminated by noise during their
acquisition and/or transmission. In particular, multichannel
imaging is  prone to quality degradation due to the imperfectness
of the sensors often operating in different spectral ranges
\cite{Abrams_M_2002_sam_rirpgmrss,Corner_B_2003_ijrs_noi_ersidm}.
In order to alleviate the influence of such disturbing artifacts
on subsequent analysis procedures, denoising appears as a crucial
initial step in multicomponent image enhancement. In this context,
attention has been paid to developing efficient   denoising
methods. Usually, the  noise removal problem is considered as a
regression problem. The challenge thus resides in  finding
realistic statistical models which lead to both efficient and
tractable denoising approaches. To this respect,  linearly
transforming the signal from the spatial  domain  to   a more
suitable  one may drastically improve the denoising performance.
 The rationale for such a transformation  is the observation that some
 representations possessing good energy concentration and
 decorrelation properties tend to simplify the statistical
 analysis of many natural images. For instance,  the Discrete
 Wavelet Transform (DWT)  constitutes  a powerful tool for image
 denoising  \cite{Daubechies_I_1992_book_ten_lw,Mallat_S_1998_book_wav_tsp}.
  The DWT, computed for each channel/component separately,   usually yields ``larger'' coefficients
for signal features and ``smaller'' ones for noise since it forms
an unconditional basis for several classes of regular
signals\cite{Donoho_D_1993_jacha_unc_bobdcse}. For monochannel
signals or images,  the seminal work of Donoho and Johnstone has
shown that a mere wavelet coefficient thresholding constitutes a
simple yet effective technique for  noise  reduction
\cite{Donoho_D_1994_biometrika_ide_saws}.   Based on Stein's
Unbiased Risk Estimator (SURE),  they have proposed the SUREshrink
technique \cite{Donoho_D_1995_jasa_ada_usws}. Subsequently,
several extensions of their work have been performed, \emph{e.g.}
in
 \cite{Nason_G_1996_jrssb_wav_scv,Wang_Y_1996_as_fun_ewslmd,Weyrich_N_1998_tip_wav_sgcvid,KRIM_DONOHO}.
 Recently, the denoising problem in the wavelet domain
 has gained more attention in the case of multichannel images. Indeed,
  the increasing need for multicomponent images in several
  applications such as medical imaging and remote sensing
  has motivated  a great interest in designing tractable
  denoising methods dedicated to this kind of images. Componentwise processing can be performed for each modality,
  but a joint denoising should be preferred in order to  exploit
  the cross-channel similarities in an efficient way \cite{Tang_J_1995_tip_non_mift}.
  The problem of a joint estimation in  the wavelet domain has been formulated
  in \cite{Fletcher_A_2003_icip_mul_et}.  More precisely,   the use of joint threshold estimators
  was investigated in two situations: overcomplete representations of a noisy image
  \cite{Fletcher_A_2003_icip_est_ebdsa}  and multiple observations of the same image \cite{Fletcher_A_2003_icip_mul_et}.
  A scale-adaptive wavelet thresholding was designed for multichannel images in the case of
  an i.i.d. (independent identically
distributed) Gaussian vector noise whose components are
independent and have the same variance
   \cite{Scheunders_p_2004_tip_wav_tmi}. In a Bayesian framework, several prior models have been considered
   such as \textit{multivariate} Bernoulli-Gaussian ones \cite{BenazzaBenyahia_A_2003_nsip_wav_bmidbgm}.
   A generalized Gaussian distribution was also considered for modelling the marginal distribution
    of each subband in each channel and a simple shrinkage was applied depending on the  local
    spectral activity \cite{Pizurica_2002_tip_joi_iismbwbid}. A vector-based least-square
approach was also investigated   in the wavelet domain
\cite{Scheunders_P_2004_icip_lea_sidcmi}. Recently, the
  application of Stein's principle \cite{Stein_C_1981_as_est_mmnd,Pesquet_J_1997_icipa_new_weid,Luisier_F_2007_tip_new_sureaidisowt}
  in the multivariate case has motivated the design
  of a nonlinear estimator 
in \cite{Benazza05}. In this paper, links existing between the proposed nonlinear estimator and Bayesian approaches were discussed.
In particular, the structure of the estimator was motivated by
  a multivariate Bernoulli-Gaussian model reflecting the sparseness of
the wavelet representation as well as the statistical dependencies
existing between the different components. 
We point out that the form of the estimator in \cite{Benazza05} is not the same as the one proposed in this paper. In particular, the estimator in \cite{Benazza05} does not involve any thresholding operation.  Moreover, the estimator does not allow to take into account spatial dependencies but only those existing between the multichannel data at a given position.\\
 In parallel to these works, the idea of performing a joint spatial denoising of
  the coefficients,  rather than using a conventional term-by-term processing, has emerged
  in statistics. This idea, stemming from an incentive for capturing   statistical
dependences between spatial neighboring wavelet coefficients, was
first investigated for single component images in both
non-Bayesian and Bayesian cases
\cite{Cai_T_2001_sankhya_inc_incwe,ABRAMOVITCH_F_2002_CSDA_EMP_BABWFE}.
A successful extension was also carried out in the case of
multichannel images by
considering hybrid (spectral and spatial) neighborhoods \cite{Chaux_C_2005_spie_blo_tmmid}. \\
In this paper, we aim at building a new estimator allowing to take
into account the various correlations existing in multichannel
image data. This estimator  also provides a unifying framework for
several denoising methods proposed in the literature. More
precisely, our  contributions are the following.\\
$\bullet$ The method applies to any
vector-valued data embedded in a multivariate Gaussian noise. As
illustrated later on, many examples of such multivariate contexts
(inter-component, spatial and inter-scale) can be found. They
naturally include  multivariate denoising obtained with vectors of
samples sharing the same spatial position in different channels.\\
$\bullet$ The estimator can be computed in any image
representation domain. For instance, in addition to wavelet
domains, usually considered in conventional denoising methods, we
propose to exploit more general frame decompositions such as the
dual-tree wavelet transform
\cite{Selesnick_I_2005_spm_dua_tcwt,Chaux_C_2006_tip_ima_adtmbwt}.\\
$\bullet$ The computation of the estimated value can be performed
with the help of various observations. Again, our method includes
most of the reported estimation methods acting in that way.
Furthermore, it offers a great flexibility in the choice of these
auxiliary data.\\ $\bullet$ The form of the proposed estimator is
quite general. More precisely, we focus on deriving thresholding
estimators including an exponent parameter and a linear part.
Optimal parameters are derived from Stein's principle.\\
$\bullet$ The denoising approach allows to handle any covariance
matrix between the multichannel noise components.

Notwithstanding its generality, the proposed approach remains
tractable and compares quite favorably with state-of-the-art
methods. The paper is organized as follows. In Section
\ref{sec:background},  we present the relevant background and
introduce notations for a general formulation of the estimator,
based on the concept of Reference Observation Vector. In Section
\ref{sec:prop_estim}, we describe the proposed  multivariate
nonlinear estimator. In Section~\ref{sec:Multi_wav_den}, we give
the specific form taken by this new estimator for  multichannel
images decomposed by a wavelet transform or  an $M$-band dual-tree
wavelet transform.
 In Section \ref{sec:simuls}, experimental results are given for
 remote sensing images showing that the proposed estimator
outperforms existing ones
and some concluding remarks are drawn in Section \ref{sec:concl}.

 Throughout this paper, the following
notations will be used: let $M$ be an integer greater than or
equal to 2, $\NM = \{0,\ldots,M-1\}$ and $\NMs =
\{1,\ldots,M-1\}$; $\ZZ$, $\RR$ and $\RR_+$ are the sets of
integers, reals and positive reals; $\lceil{.}\rceil$ denotes
rounding towards the immediate upper integer. Besides,
$\widehat{a}$ denotes the Fourier transform of a function $a$,
$(\delta_m)_{m\in \ZZ}$ is the Kronecker sequence (equal to 1 if
$m=0$ and 0 otherwise), $\left( f\right)_+ = f$ if $f> 0$ and $0$
otherwise and $\mathds{1}\{A\} \,=\, 1$ if condition $A$ is true
and 0 otherwise.

\section{Background}
\label{sec:background}
\subsection{General formulation of the multichannel estimator}
In multisensor imaging,  $B$ vectors
of  observed data samples $(r^{(1)}({\mathbf{k}}))_{\kk \in \KK}$,
\ldots, $(r^{(B)}({\mathbf{k}}))_{\kk \in \KK}$, are provided
where $B$ is the number of effective sensors and $\KK$ is a set
of spatial indices ($\KK\subset \ZZ^2$). Generally, these data
correspond to noisy realizations of $B$ unknown signals
$(s^{(1)}({\mathbf{k}}))_{\kk \in \KK}$, \ldots,
$(s^{(B)}({\mathbf{k}}))_{\kk \in \KK}$, respectively.
Subsequently, our task will consist in devising methods to reduce
the noise present in the observations.
Two alternatives can be envisaged in this context. On the one
hand, a monochannel approach  builds an estimator
$\est{s}{(b)}(\mathbf{k})$ of $s^{(b)}(\mathbf{k})$ only from the
observations $(r^{(b)}({\mathbf{k}}))_{\kk \in \KK}$, for each
channel $b\in \{1,\ldots,B\}$. On the other hand, a multivariate
technique attempts to estimate $s^{(b)}(\mathbf{k})$  by
accounting not only for the individual data set
$\{r^{(b)}({\mathbf{k}})\}_{\kk \in \KK}$, but also for the
remaining ones $\{r^{(1)}({\mathbf{k}})\}_{\kk \in \KK}$, \ldots,
$\{r^{(b-1)}({\mathbf{k}})\}_{\kk \in \KK}$,
$\{r^{(b+1)}({\mathbf{k}})\}_{\kk \in \KK}$, \ldots,
$\{r^{(B)}({\mathbf{k}})\}_{\kk \in \KK}$.

Thus, one of the simplest relevant denoising approach consists in
calculating the estimated value $\est{s}{(b)}(\mathbf{k})$ of
$s(\mathbf{k})$ as
\begin{equation}
\est{s}{(b)}(\mathbf{k}) =
f(r^{(b)}(\mathbf{k}))
\end{equation}
where $f$ is a scalar function defined on the real line. For
instance, a shrinkage function can be used, possibly involving
some threshold value. Such a technique is commonly used in
regression, when outliers have to be removed in order to improve
the representativity of the fit
\cite{Antoniadis_A_2002_statneer_wtcngn}. Although
$r^{(b)}(\mathbf{k})$ does not necessarily depend on other
observed samples, for structured signal or image analysis,
neighboring samples often present some correlations. Consequently,
an improvement can be expected if    $\est{s}{(b)}(\mathbf{k})$ is
calculated with the help of a subset
$\mathcal{R}_\textrm{ref}^{(b)}(\mathbf{k})$ of observed sample
locations. Average  or median filtering \cite[p.
243--245]{Maitre_H_2003_book_tra_i} are examples where
the estimated sample depends on its neighborhood. 
 As a result, a more general estimation rule is:
\begin{equation}
\est{s}{(b)}(\mathbf{k}) = f\big((r^{(b)}({\mathbf{k}}^\prime))_{{\mathbf{k}}^\prime\in \mathcal{R}_\textrm{ref}^{(b)}({\mathbf{k}})}\big).
\end{equation}
With underlying Markovian assumptions, the context set
$\{r^{(b)}({\mathbf{k}}^\prime)\}_{{\mathbf{k}}^\prime\in
\mathcal{R}_\textrm{ref}^{(b)}({\mathbf{k}})}$ can be restricted
to a limited number of values around the sample location $\kk$.
These values can be gathered in a vector $\bbf{r}^{(b)}(\kk)$
which will be designated as the \textit{Reference Observation
Vector} (ROV). We have then
\begin{equation}
\est{s}{(b)}(\mathbf{k}) = f(\bbf{r}^{(b)}(\kk)).
\label{eq:criterion}
\end{equation}
The multivariate case can also be described by such a formula if
we allow  the ROV to contain additional samples from the remaining
channels in order to exploit the inter-component statistical
dependencies.

Another degree of freedom lies in the choice of a suitable domain
for data representation. While virtually any  transform can be
chosen, special attention has been paid to multiscale transforms.
For example, if a decomposition onto an $M$-band wavelet basis
($M\ge 2$) \cite{Steffen_P_1993_tsp_the_rmbwb} is performed, the
observed images are represented by coefficients
$r_{j,\mathbf{m}}^{(b)}({\mathbf{k}})$ defined at resolution level
$j\ge 1$  and  subband index $\mathbf{m}\in \NM^2$ and the
corresponding ROV will be denoted
$\overline{{\mathbf{r}}}_{j,\mathbf{m}}^{(b)}({\mathbf{k}})$.
Since the noise is usually less correlated than the data, the DWT
is applied in order to provide a sparser representation of the
data of interest, before further analysis
\cite{Daubechies_I_1992_book_ten_lw,Mallat_S_1998_book_wav_tsp}.
The goal becomes to generate estimates
$\estw{s}{(b)}_{j,\mathbf{m}}({\mathbf{k}})$ of the unknown
wavelet coefficients $s_{j,\mathbf{m}}^{(b)}({\mathbf{k}})$ of the
original images:
\begin{equation}
\estw{s}{(b)}_{j,\mathbf{m}}({\mathbf{k}}) =
f(\bbf{r}_{j,\mathbf{m}}^{(b)}(\kk)).
\label{eq:criterion_ROV_transf}
\end{equation}
Then, the inverse DWT is applied to the estimated coefficients in order to
reconstruct the estimated signal $\est{s}{(b)}(\mathbf{k})$ in the spatial domain.
In the literature concerning denoising, two  key issues have been addressed.
The first one lies in the definition of the ROV.
 The second one concerns the choice of an appropriate function $f$ or, in other words, a suitable
 expression of the estimator. In the next subsection, we give a brief overview of the main ROVs proposed until now.

\subsection{Reported ROVs in the DWT domain}
\label{sec:gen_neighborhood} Popular componentwise methods
operating in the  DWT domain are Visushrink
\cite{Donoho_D_1993_jacha_unc_bobdcse} and SUREshrink
\cite{Donoho_D_1995_jasa_ada_usws}. They both employ a very basic
ROV reduced to a scalar value:
\begin{equation}
\bbf{r}^{(b)}_{j,\mathbf{m}}(\kk)= r_{j,\mathbf{m}}^{(b)}({\mathbf{k}}).
\label{eq:scalar_ROV}
\end{equation}
Similarly to what can be done in the spatial domain, the wavelet
coefficients can also be processed  by \textit{block} rather than
individually, again in a mono-channel way
\cite{Hall_P_1997_statist-comput_num_pbtwe,Hall_P_1998_annals-statistics_blo_trcekwm,Hall_P_1999_stat-sinica_min_obtwe,Cai_T_2001_sankhya_inc_incwe,ABRAMOVITCH_F_2002_CSDA_EMP_BABWFE}.
The main motivation for this technique is to exploit the spatial
similarities between neighboring coefficients in a given subband.
The introduction of $d-1$ spatial neighbors ${\mathbf{k}}_1$
,\dots, ${\mathbf{k}}_{d-1}$ of the current sample indexed by
${\mathbf{k}}$ in the ROV allows to take into account the spatial
dependencies:
\begin{equation}
\overline{{\mathbf{r}}}_{j,\mathbf{m}}^{(b)}({\mathbf{k}})=[
r_{j,\mathbf{m}}^{(b)}({\mathbf{k}}),
r_{j,\mathbf{m}}^{(b)}({\mathbf{k}}_1), \ldots,
r_{j,\mathbf{m}}^{(b)}({\mathbf{k}}_{d-1})]^\top.
 \label{eq:spatial_ref}
\end{equation}

For higher dimensional data, the ROV may also consist of
coefficients sharing similar orientations, possibly within
different scales \cite{Portilla_J_2003_tip_ima_dsmgwd}. Another
generalization of the scalar case takes into account the
inter-scale similarities between the current coefficient and the
homologous ones defined at other scales. Based on empirical
observations in image compression
\cite{Shapiro_J_1993_tsp_emb_iczwc}, it has been proposed to use
the current coefficient ancestors at  coarser scales $j+1$, $j+2$,
\ldots, $j_\mathrm{m}$ eventually up to the coarsest level $J$
\cite{romberg_j_2001_tip_bay_tsimwdhmm,sendur_l_2002_tsp_biv_sfwbdeid,benazzabenyahia_a_2004_eusipco_int_mmapemi}: the ROV $\overline{{\mathbf{r}}}_{j,\mathbf{m}}^{(b)}({\mathbf{k}})$ thus includes the corresponding $j_{\mathrm{m}}-j+1$ coefficients at location ${\mathbf{k}}$, in subband $\mathbf{m}$, at resolution level $j$.

In the case of multicomponent data, additional samples borrowed
from the different channels can be included in the ROVs, as shown
in
\cite{Portilla_J_2003_tip_ima_dsmgwd,Pizurica_A_2006_tip_est_ppsimsmid}
for color image as well as for multispectral image denoising.
Basically, the \textit{inter-component} correlations can be taken
into account through the following ROV \cite{Benazza05}:
\begin{equation}
\overline{{\mathbf{r}}}_{j,\mathbf{m}}^{(b)}({\mathbf{k}})=[
r_{j,\mathbf{m}}^{(1)}({\mathbf{k}}),\ldots,
r_{j,\mathbf{m}}^{(B)}({\mathbf{k}})]^\top.
\label{eq:spectr_ref}
\end{equation}
Such a ROV includes all the coefficients of all channels at the
same spatial location, in the same subband $\mathbf{m}$ and at the
same resolution level $j$. In \cite{Chaux_C_2005_spie_blo_tmmid},
a  more sophisticated multicomponent ROV $\overline{{\mathbf{r}}}_{j,\mathbf{m}}^{(b)}({\mathbf{k}})$ has been defined which
combines both spatial \textit{and} multichannel neighbors.

As particular cases, such an ROV encompasses the ROV in
\eqref{eq:spectr_ref} 
and, also the ROV in  \eqref{eq:spatial_ref}.
In addition, the ROV may include coefficients from different subbands.

A final potential extension of the ROVs is related to the choice
of the transform. Indeed, it has been long observed that a
decomposition onto a wavelet basis suffers from a lack of
shift-invariance as well as a poor directionality, resulting in
denoising artifacts at low signal to noise ratios. A simple way
for alleviating these problems is to use a frame decomposition
built from a union of wavelet bases. In particular, a number of
papers
\cite{Coifman_R_1995_was_tra_id,Nason_G_1995_was_sta_wtsa,Pesquet_J_1996_tsp_tim_iowr}
have demonstrated significant improvements in scalar
shrinkage when resorting to a translation-invariant wavelet representation. 
The latter can be viewed as a decomposition onto a union of
shifted versions of a unique wavelet basis.  $M$-band dual-tree
wavelet decompositions \cite{Chaux_C_2006_tip_ima_adtmbwt}
constitute another example of a union of 2 (resp. 4) wavelet bases
in the real (resp. complex) case. The corresponding mother
wavelets are then derived from the first one by Hilbert
transforms, which results in an improved directional analysis. For
such frame decompositions, one can extend the notion of ROV to
include samples produced by the different wavelet basis
decompositions operating in parallel.
These facts will be further developed %
to motivate the application of the general estimator proposed in this paper to an $M$-band dual-tree
 wavelet frame \cite{Chaux_C_2006_tip_ima_adtmbwt}.

\subsection{A unifying framework for shrinkage functions}
\label{sec:unify_framew} In the aforementioned works, the
estimation is often performed by shrinkage, so exploiting the
sparseness of the representation. The most well-known method was
proposed in the pioneering works of Donoho and Johnstone
\cite{Donoho_D_1993_jacha_unc_bobdcse}.
The estimating function $f$ is then given by
\begin{equation}
f(r_{j,{\mathbf{m}}}^{(b)}({\mathbf{k}})) =
\mathrm{sign}(r_{j,{\mathbf{m}}}^{(b)}({\mathbf{k}}))\max\{|r_{j,{\mathbf{m}}}^{(b)}({\mathbf{k}})|-\lambda,0\}
\label{eq:shrinkage_soft_plus1}
\end{equation}
for a soft thresholding with threshold value $\lambda \ge 0$,  where
$\mathrm{sign(\cdot)}$ is the signum function.
Equivalently, by using the ROV in (\ref{eq:scalar_ROV}),
the estimating function can be expressed as
\begin{equation}
f(\bbf{r}_{j,{\mathbf{m}}}^{(b)}({\kk})) =
\Bigg(\frac{|\bbf{r}_{j,{\mathbf{m}}}^{(b)}(\kk)|-\lambda}{|\bbf{r}_{j,{\mathbf{m}}}^{(b)}{\kk})|}\Bigg)_+
r_{j,{\mathbf{m}}}^{(b)}({\mathbf{k}}).
\label{eq:shrinkage_soft_plus}
\end{equation}
Some works \cite{Gao_H_1998_jcgs_wav_sdnng} have focused on the
improvement of the scalar shrinkage rule, yielding for instance
smoother functions such as the garrote shrinkage based on
\cite{Breitman_L_1995_technometrics_bet_srng}, which is defined
as:
\begin{equation}
f(\bbf{r}_{j,{\mathbf{m}}}^{(b)}({\kk})) =
\Bigg(\frac{|\bbf{r}_{j,{\mathbf{m}}}^{(b)}(\kk)|^2-\lambda}{|\bbf{r}_{j,{\mathbf{m}}}^{(b)}(\kk)|^2}\Bigg)_+\
r_{j,{\mathbf{m}}}^{(b)}({\mathbf{k}}).
\label{eq:shrinkage_garrote}
\end{equation}
Several  authors have proposed vector-like generalizations to the
scalar shrinkage. Cai and Silverman
\cite{Cai_T_2001_sankhya_inc_incwe}, have proposed a block
estimator which takes into account information on neighboring
coefficients in each subband, as expressed as in
\eqref{eq:spatial_ref}. This estimator dominates the maximum likelihood estimator 
when the block size is greater than 2. This method, named ``NeighBlock'',
consists in applying the following shrinkage rule:
\begin{equation}
\estw{\vect{s}}{(b)}_{j,{\mathbf{m}}}({\mathbf{k}}) =
\left(\frac{\|\bbf{r}_{j,{\mathbf{m}}}^{(b)}(\kk)\|^2-\bar{\lambda}d\sigma^2}{\|\bbf{r}_{j,{\mathbf{m}}}^{(b)}(\kk)\|^2}\right)_+\
\vect{r}_{j,{\mathbf{m}}}^{(b)}({\mathbf{k}})
\label{eq:neighblock}
\end{equation}
where  $\bar{\lambda} > 0$, $d$ is the number of components in the
ROV, $\vect{r}_{j,{\mathbf{m}}}^{(b)}({\mathbf{k}})$ is a subpart
of the ROV, $\estw{\vect{s}}{(b)}_{j,{\mathbf{m}}}({\mathbf{k}})$
is the associated vector of estimated values, $\|.\|$ denotes the
classical Euclidean norm of $\mathbb{R}^{d}$ and $\sigma^2$
denotes the noise variance. Such a function is clearly reminiscent
of the scalar garrote shrinkage defined in
\eqref{eq:shrinkage_garrote}. Based on an asymptotic minimax
study, Cai and Silverman suggested appropriate values for
$\bar{\lambda}$ and $d$. They considered both overlapping and
non-overlapping variants of this approach. In particular, the
so-called ``NeighCoeff'' method corresponds to  the case when
$\estw{\vect{s}}{(b)}_{j,\mathbf{m}}({\mathbf{k}})$ reduces to a
scalar estimate. Then, the corresponding estimating function is:
\begin{equation}
f(\bbf{r}_{j,{\mathbf{m}}}^{(b)}({\kk})) =
\Bigg(\frac{\|\bbf{r}_{j,{\mathbf{m}}}^{(b)}(\kk)\|^2-\bar{\lambda}d\sigma^2}{\|\bbf{r}_{j,{\mathbf{m}}}^{(b)}(\kk)\|^2}\Bigg)_+\
r_{j,{\mathbf{m}}}^{(b)}({\mathbf{k}}).
\label{eq:neighcoef}
\end{equation}

In the meantime,  \c{S}endur and Selesnick
\cite{Sendur_L_2002_spl_biv_slve} introduced a Bayesian approach
allowing to model inter-scale dependencies between two consecutive levels. These authors
con\-se\-quent\-ly formulated the problem in the 2-band wavelet
domain. In their approach, the ROV is given by
$\overline{{\mathbf{r}}}_{j,\mathbf{m}}^{(b)}({\mathbf{k}})=[
r_{j,\mathbf{m}}^{(b)}({\mathbf{k}}),r_{j+1,\mathbf{m}}^{(b)}(\lceil{\frac{\mathbf{k}}{2}\rceil})]^\top$,
$r_{j+1,\mathbf{m}}^{(b)}(\lceil{\frac{\mathbf{k}}{2}\rceil})$
being the ``parent'' of $r_{j,\mathbf{m}}^{(b)}({\mathbf{k}})$ (at
the next coarser resolution). By considering as a prior model the
non-Gaussian bivariate probability density function
\begin{equation}
p(s_{j,\mathbf{m}}^{(b)}({\mathbf{k}}),s_{j+1,\mathbf{m}}^{(b)}(\lceil{\frac{\mathbf{k}}{2}\rceil}))
\propto \exp\Big(-\frac{\sqrt{3}}{\sigma_s}\sqrt{
\big|s_{j,\mathbf{m}}^{(b)}({\mathbf{k}})\big|^2+\big|s_{j+1,\mathbf{m}}^{(b)}(\lceil{\frac{\mathbf{k}}{2}\rceil})\big|^2}\;\Big), \qquad \sigma_s > 0
\end{equation}
the following Maximum A Posteriori (MAP) estimator was derived:
\begin{equation}
f(\bbf{r}_{j,{\mathbf{m}}}^{(b)}({\kk}))= \left(
\frac{\|\bbf{r}_{j,{\mathbf{m}}}^{(b)}(\kk)\|-\frac{\sqrt{3}\sigma^2}{\sigma_s}}{\|\bbf{r}_{j,{\mathbf{m}}}^{(b)}(\kk)\|}\right)_+r_{j,\mathbf{m}}^{(b)}({\mathbf{k}})
\label{eq:bivariate}
\end{equation}
where the noise variance is again denoted by $\sigma^2$.\\
More recently, in the context of signal restoration problems,
Combettes and Wajs \cite{Combettes_P_2005_siam-mms_sig_rpfbs} have
studied the properties of proximity operators corresponding to the
solutions of some convex regularization problems. In particular, 
an interpretation
of one of their results is the following. Let us adopt a Bayesian formulation
by assuming that the vector $\bbf{r}_{j,{\mathbf{m}}}^{(b)}({\kk})$
is a noisy observation of the vector $\vect{s}^{(b)}_{j,{\mathbf{m}}}({\kk})$
of multichannel coefficients at location $\kk$, embedded in white Gaussian noise
with variance $\sigma^2$.
Further assume that the vectors $\vect{s}^{(b)}_{j,{\mathbf{m}}}({\kk})$
are independent of the noise, mutually independent and have a prior distribution proportional
to $\exp(-\overline{\lambda} \|\cdot\|)$ with $\overline{\lambda} > 0$. The MAP estimation
of $\vect{s}^{(b)}_{j,{\mathbf{m}}}$ is found by solving
the optimization problem:
\begin{equation}
\min_{\vect{u}\in \mathbb{R}^B} \overline{\lambda}
\| \vect{u}\|+\frac{1}{2\sigma^2}
\|\vect{u}-\bbf{r}_{j,{\mathbf{m}}}^{(b)}({\kk})\|^2.
\end{equation}
It is shown in \cite{Combettes_P_2005_siam-mms_sig_rpfbs}
that the minimizer of the MAP criterion is
\begin{equation}
\estw{\vect{s}}{(b)}_{j,{\mathbf{m}}} =
\Big(\frac{\|\bbf{r}_{j,{\mathbf{m}}}^{(b)}({\kk})\|-
\overline{\lambda}\sigma^2}{\|\bbf{r}_{j,{\mathbf{m}}}^{(b)}({\kk})\|}\Big)_+\bbf{r}_{j,{\mathbf{m}}}^{(b)}({\kk}).
\end{equation}

The three previous block-thresholding estimators have been derived
from different perspectives and they have also been applied in
different ways. However, it is possible to describe them through a
general shrinkage factor
$\eta_{\lambda}(\|\vect{\bar{r}}_{j,\mathbf{m}}^{(b)}\|^\beta)$,
where
\begin{equation}
\forall \tau\in \RR_+,\qquad
\eta_{\lambda}(\tau)=\Big(\frac{\tau-\lambda}{\tau}\Big)_+
\label{eq:rule}
\end{equation}
and $\beta>0$ and $\lambda\ge 0$ take specific values in each of the aforementioned block estimators.
We also remark that
this generalized shrinkage obviously encompasses the soft and
garrote thresholdings provided in \eqref{eq:shrinkage_soft_plus} and \eqref{eq:shrinkage_garrote}.

\section{Proposed nonlinear estimator}
\label{sec:prop_estim}
\subsection{Notations}
We will now propose a more general adaptive estimator that can be
applied in any representation domain. We will therefore drop the
indices $j$ and ${\mathbf{m}}$  and we will consider the general
situation where an observation sequence $(\bbf{r}(\kk))_{\kk \in
\ZZ^2}$ of $d$-dimensional real-valued vectors ($d\in \NN$, $d >
1$) is defined as
\begin{equation}
\forall \kk \in \ZZ^2,\qquad \bbf{r}(\kk)=\bbf{s}(\kk)+\bbf{n}(\kk),
\label{eq:modelgen}
\end{equation}
where $(\bbf{n}(\kk))_{\kk\in\ZZ^2}$ is a
$\mathcal{N}(\mathbf{0},\GGamma^{(\bbf{n})})$ noise
and $(\bbf{s}(\kk))_{\kk\in\ZZ^2}$ is an identically distributed
second-order random sequence which is independent  of
$(\bbf{n}(\kk))_{\kk\in \ZZ^2}$. We will assume that  the
covariance matrix $\GGamma^{(\bbf{n})}$ is invertible.
These random vectors are decomposed as
\begin{equation}
\bbf{r}({\mathbf{k}}) =
\begin{bmatrix}
r({\mathbf{k}})\\
\tilde{\mathbf{r}}({\mathbf{k}})
\end{bmatrix},\qquad
\bbf{s}({\mathbf{k}}) =
\begin{bmatrix}
s({\mathbf{k}})\\
\tilde{\mathbf{s}}({\mathbf{k}})
\end{bmatrix},\qquad
\bbf{n}({\mathbf{k}}) =
\begin{bmatrix}
n({\mathbf{k}})\\
\tilde{\mathbf{n}}({\mathbf{k}})
\end{bmatrix}
\label{eq:decompos}
\end{equation}
where $r({\mathbf{k}})$, $s({\mathbf{k}})$ and $n({\mathbf{k}})$
are scalar random variables. We aim at estimating the first
component  $s(\kk)$ of the vector $\bbf{s}(\kk)$ using an
observation sequence $(\bbf{r}(\kk))_{\kk \in \KK}$ where $\KK$ is
a finite subset of $\ZZ^2$. We recall that, although
\eqref{eq:modelgen} does not introduce an explicit dependence
between $s(\kk)$ and the vector $\tilde{\mathbf{r}}(\kk)$ of the
last $d-1$ components of $\bbf{r}(\kk)$, such a statistical
dependence may exist, due to the dependence between the components
of  $\bbf{s}(\kk)$ themselves. The estimated sequence will be
denoted by $(\ests{s}(\kk))_{\kk \in \KK}$.

\subsection{Form of the adaptive estimator}
In order to gain more flexibility in the denoising procedure,
the following generalized form of shrinkage estimate will be considered:
\begin{equation}
\ests{s}(\kk) = \eta_{\lambda}(\|\bbf{r}(\kk)\|_{\LLa}^\beta) \;
\vect{q}^{\top}\bbf{r}(\kk),\label{estimator}
\end{equation}
where
the function $\eta_{\lambda}(\cdot)$ is given by \eqref{eq:rule}
with $\lambda \ge 0$, $\beta >0$ and $\vect{q}\in \RR^d$.
The vector $\vect{q}$ corresponds to a linear
parameter. We notice, in particular, that if the threshold value
$\lambda$ is set to zero, the considered estimator reduces to $
\estw{s}{(b)}\;\;({\mathbf{k}})=
\vect{q}^{\top}\vect{\bar{r}}({\mathbf{k}}).$ This shows that
linear estimators constitute a subset of the considered class of
estimators. In addition, by an appropriate choice of the vector
$\vect{q}$, estimators consisting of a preliminary decorrelation
of the data  followed by a thresholding step also appear as
special cases of the proposed estimator. Note that, in
conventional multichannel data analysis, it is customary to
decorrelate the data before processing. The most common examples
are fixed channel conversions (like those from stereo to mono  in
sound processing or from RGB to luminance/chrominance components
in color image or video processing). When the data modalities are
less standardized (for instance in satellite imaging),  adaptive
methods such as the Karhunen-Lo{\`e}ve transform or Independent
Component Analysis (ICA) \cite{Cardoso_J_1993_ieepf_bli_bngs} can
be used. The latter adaptive transforms can also be performed in
the transformed
domain, \emph{e.g.} in each subband.\\
Furthermore, in order to limit the computational complexity in the
implementation of the estimator, it can be useful to constrain the
vector $\vect{q}$ to belong to some vector subspace of reduced
dimension $d'\le d$. Let $\mathbf{P}\in \RR^{d\times d'}$ be the
matrix whose column vectors form a basis of this subspace. We have
then $\vect{q} = \mathbf{P}\aaa$ where $\vect{a}\in \RR^{d'}$.
As a simple example, by choosing
\[
\mathbf{P}
= \begin{bmatrix}
\mathbf{I}_{d'}\\
\mathbf{O}
\end{bmatrix}
\]
where $\mathbf{I}_{d'}$ denotes the identity matrix of size $d'\times d'$,
we see that we only introduce in the estimator a linear combination of the first $d'$ components of
the vector $\overline{\mathbf{r}}(\mathbf{k})$.
In summary, the proposed form of the estimator is parameterized by
$\lambda$, $\beta$ and $\aaa$
for a given choice of $\mathbf{P}$.  \\

Our objective is to find the  optimal parameters
that minimize the quadratic risk defined as
$R(\lambda,\beta,\aaa) = \E[|s(\kk) - \ests{s}(\kk)|^2]$, for a predefined value of 
$\mathbf{P}$.
It is easy to show that the risk reads:
\begin{align}
R(\lambda,\beta,\aaa) &= \E[|s(\kk) - \ests{s} (\kk)|^2] \nonumber\\
&= \E[|s(\kk)|^2]+\E[|\eta_{\lambda}(\|  \bbf{r}(\kk) \|_{\LLa}^\beta)\aaa^{\top}\mathbf{P}^\top \bbf{r}(\kk)  |^2]-2 \E[  \eta_{\lambda}(\|  \bbf{r}(\kk) \|_{\LLa}^\beta)\aaa^{\top}\mathbf{P}^\top \bbf{r}(\kk)  s(\kk)].
\end{align}
 The minimization of the risk is not obvious for any observation model.
 Indeed, since the $s({\mathbf{k}})$ are unknown, it seems impossible to express the rightmost term
 $\E[  \eta_{\lambda}(\|  \bbf{r}(\kk) \|_{\LLa}^\beta)\aaa^{\top}\mathbf{P}^\top \bbf{r}(\kk)  s(\kk)]$.
 However, in the case of a \textit{Gaussian} noise,  it  is possible to apply an extension of Stein's
 principle \cite{Stein_C_1981_as_est_mmnd} for deriving an explicit expression. In the next subsection, we will state
 and prove such extended Stein's formula.

\subsection{Stein's formula}
\begin{prop}
 \label{prop:Steinextended}
Let $f\;:\;\mathbb{R}^d \to \mathbb{R}$ be a continuous,
almost everywhere differentiable function such that:
\begin{align}
&\forall \boldsymbol{\theta} \in \RR^d, \, \lim_{\|\mathbf{t}\|\to +\infty} f(\mathbf{t}) \exp\Big(-\frac{(\mathbf{t}-\boldsymbol{\theta})^\top ({\GGamma^{(\bbf{n})}})^{-1}(\mathbf{t}-\boldsymbol{\theta})}{2}\Big)=0;\\
&\E[|f(\bbf{r}(\kk))|^2]<+\infty \quad \mbox{and} \quad  \E\big[\big\|\frac{\partial f(\bbf{r}(\kk))}{\partial \bbf{r}(\kk)}\big\|\big]<+\infty.
\end{align}
Then,
\begin{equation}
\E[f(\bbf{r}(\kk))s(\kk)] =
\E[f(\bbf{r}(\kk))r(\kk)]-
 \E\Big[\frac{\partial f(\bbf{r}(\kk))}{\partial \bbf{r}(\kk)}\Big]^\top\E[\bbf{n}
 n].
\label{eq:stein}
\end{equation}
\end{prop}

\begin{proof}
Let $\vect{T}\;:\;\mathbb{R}^d \to \mathbb{R}^d$ be a continuous,
almost everywhere differentiable function such that
\begin{align}
&\forall \boldsymbol{\theta} \in \RR^d, \, \lim_{\|\mathbf{t}\|\to +\infty} \vect{T}(\mathbf{t}) \exp\Big(-\frac{(\mathbf{t}-\boldsymbol{\theta})^\top ({\GGamma^{(\bbf{n})}})^{-1}(\mathbf{t}-\boldsymbol{\theta})}{2}\Big)=\vect{0};\\
&\E[\|\vect{T}(\bbf{r}(\kk))\|^2]<+\infty \quad \mbox{and} \quad  \E\Big[\big\|\frac{\partial \vect{T}(\bbf{r}(\kk))}{\partial \bbf{r}^\top(\kk)}\big\|_{\mathrm{F}}\Big]<+\infty.
\end{align}
where $\|\cdot\|_{\mathrm{F}}$ is the Frobenius norm. In this
multivariate context, Stein's principle
\cite{Stein_C_1981_as_est_mmnd} can be expressed as
\begin{equation}
\E[\vect{T}(\bbf{r}(\kk))\bbf{s}^\top(\kk)] =
\E[\vect{T}(\bbf{r}(\kk))\bbf{r}^\top(\kk)]-
\E[\frac{\partial \vect{T}(\bbf{r}(\kk))}{\partial \bbf{r}^\top(\kk)}] \GGamma^{(\bbf{n})}.
\end{equation}
Eq. \eqref{eq:stein} follows by choosing
$ \vect{T}\;:\;\vect{t} \mapsto [f(\vect{t}),0,\ldots,0]^\top $
and focusing on the top-left element of matrix
$\E[\vect{T}(\bbf{r}(\kk))\bbf{s}^\top(\kk)]$.
\end{proof}

\subsection{Risk expression}

We define the function $f\,:\,\vect{u}\mapsto  \eta_{\lambda}(\|
\vect{u}\|_{\LLa}^\beta) \; \aaa^{\top}\mathbf{P}^\top\vect{u}$.
It is easy  to check that this function $f$ satisfies the
conditions of Prop. \ref{prop:Steinextended}. Consequently, the
last term can be calculated thanks to \eqref{eq:stein}. This
yields
\begin{equation}
\E[s(\kk) f(\bbf{r}(\kk))]=\E[r(\kk) f(\bbf{r}(\kk))]
-\E\Big[\frac{\partial f\big(\bbf{r}(\kk)\big)}{\partial \bbf{r}(\kk)} \Big]^{\top}
\GGamma^{(\bbf{n},n)}
\label{eq:riskstein}
\end{equation}
where $\GGamma^{(\bbf{n},n)}=\E[\bbf{n}(\kk)n(\kk)]$.
We then have
\begin{align}
\frac{\partial f\big(\bbf{r}(\kk)\big)}{\partial
\bbf{r}(\kk)}
&= \aaa^{\top}\mathbf{P}^{\top}\bbf{r}(\kk)\frac{\partial \eta_{\lambda}(\|\bbf{r}(\kk)\|_{\LLa}^\beta)}{\partial \bbf{r}(\kk)} + \eta_{\lambda}(\|\bbf{r}(\kk)\|_{\LLa}^\beta)\frac{\partial \aaa^{\top}\mathbf{P}^{\top}\bbf{r}(\kk)}{\partial \bbf{r}(\kk)}  \nonumber \\
&= \frac{\lambda
\beta}{\|\bbf{r}(\kk)\|_{\LLa}^{\beta+1}}{\mathds{1}}\{\parallel
\overline{{\mathbf{r}}}(\kk)\|_{\LLa}^\beta>\lambda\}\frac{\partial\|\bbf{r}(\kk)\|_{\LLa}}{\partial
\bbf{r}(\kk)}\,\bbf{r}^\top(\kk) \mathbf{P} \aaa
 + \eta_{\lambda}(\|\bbf{r}(\kk)\|_{\LLa}^\beta) \mathbf{P} \aaa   \nonumber \\
&=
 \eta_{\lambda}(\|\bbf{r}(\kk)\|_{\LLa}^\beta) \mathbf{P} \aaa + \lambda  \LLa \bbf{r}(\kk) \boldsymbol{\xi}^\top(\kk)\mathbf{P}\aaa
\end{align}
where $\boldsymbol{\xi}(\kk) ={\mathds{1}}\{\parallel
\overline{{\mathbf{r}}}(\kk)\|_{\LLa}^\beta>\lambda\}\frac{\beta}{\parallel
\overline{{\mathbf{r}}}(\kk)\|_{\LLa}^{\beta+2}}\,\bbf{r}(\kk).$
This leads to the following expression of the risk:
\begin{equation}
R(\lambda,\beta,\aaa) =
\E[|r(\kk)-f(\bbf{r}(\kk))|^2]
+2 \E[\eta_{\lambda}(\|\bbf{r}(\kk)\|_{\LLa}^\beta)] \aaa^\top\mathbf{P}^\top \GGamma^{(\bbf{n},n)} + 2 \lambda\aaa^\top\mathbf{P}^\top \E[ \boldsymbol{\xi}(\kk)\bbf{r}^\top(\kk)]\LLa  \GGamma^{(\bbf{n},n)}-\sigma^2
\label{eq:risk}
\end{equation}
where $\sigma^2 = \E[|n(\kk)|^2]$.

We will now look for parameters $\lambda$, $\beta$ and $\aaa$
  that minimize the risk expression \eqref{eq:risk} for a given choice of $\mathbf{P}$.

\subsection{Determination of the parameter $\aaa$}
\label{sec:estim_q} We first aim at calculating the  value of $\aaa$
that minimizes the risk \eqref{eq:risk}. By noticing that the risk is a quadratic convex function of $\aaa$,
the minimization can be performed by
differentiating
w.r.t. $\aaa$ and then finding $\aaa^*(\lambda,\beta)$ such that
$\partial R/\partial \aaa\big(\lambda,\beta,\aaa^*(\lambda,\beta)\big)=0$.
It readily follows that
\begin{multline}
\aaa^*(\lambda,\beta)= \big(\mathbf{P}^\top\E[\eta_\lambda^2(\|\bbf{r}(\kk)\|_{\LLa}^\beta)\bbf{r}(\kk)\bbf{r}^\top(\kk)]\mathbf{P}\big)^{-1}\, \mathbf{P}^\top\big(\E[\eta_\lambda(\|\bbf{r}(\kk)\|_{\LLa}^\beta)r(\kk)\bbf{r}(\kk) ]\\
-\E[\eta_\lambda(\|\bbf{r}(\kk)\|_{\LLa}^\beta)]{\GGamma}^{(\bbf{n},n)} - \lambda\, \E[{\boldsymbol{\xi}}(\kk) \bbf{r}^{\top}(\kk)]\LLa{\GGamma}^{(\bbf{n},n)}\big).
\label{eq:opta}
\end{multline}

\subsection{Determination of the parameters $\lambda$ and $\beta$}
\label{sec:estim_beta_lambda} Starting from \eqref{eq:risk}, the
risk $R(\lambda,\beta,\aaa)$ can be re-expressed as
$R(\lambda,\beta,\aaa)=\E[\rho_{\lambda,\beta,\aaa}(\kk)]$ where
\begin{equation}
\rho_{\lambda,\beta,\aaa}(\kk)= \alpha_2(\kk) \lambda^2 +
\alpha_1(\kk) \lambda + \alpha_0(\kk)
\end{equation}
and
\begin{align*}
\alpha_0(\kk) & = r^2(\kk)-\sigma^2+ {\mathds{1}}\{\parallel
\overline{{\mathbf{r}}}(\kk)\|_{\LLa}^\beta>\lambda\}
\aaa^{\top}\mathbf{P}^{\top}\Big(2\,\GGamma^{(\bbf{n},n)}+\big(\aaa^{\top}\mathbf{P}^{\top}\bbf{r}(\kk)-2\,r(\kk)\big)\,\bbf{r}(\kk)\Big)\nonumber\\
 \alpha_1(\kk) & =
 2\aaa^{\top}\mathbf{P}^{\top}\Big(\frac{\big(r(\kk)-\aaa^{\top}\mathbf{P}^{\top}\bbf{r}(\kk)\big)\,\bbf{r}(\kk)
 - \GGamma^{(\bbf{n},n)}}{\|\bbf{r}(\kk)\|_{\LLa}^\beta}
 +\beta\,\bbf{r}(\kk)\frac{\bbf{r}^{\top}(\kk)\LLa
 \GGamma^{(\bbf{n},n)}}{\|\bbf{r}(\kk)
 \|_{\LLa}^{\beta+2}}\Big){\mathds{1}}\{\parallel \overline{{\mathbf{r}}}(\kk)\|_{\LLa}^\beta>\lambda\}
\nonumber\\
\alpha_2(\kk) & = {\mathds{1}}\{\parallel
\overline{{\mathbf{r}}}(\kk)\|_{\LLa}^\beta>\lambda\}
\frac{\big(\aaa^{\top}\mathbf{P}^{\top}\bbf{r}(\kk)\big)^2}{\|\bbf{r}(\kk)\|_{\LLa}^{2\beta}}.
\end{align*}
In practice, under standard mixing assumptions for
$(\bbf{n}(\kk))_{k\in \ZZ^2}$ and $(\bbf{s}(\kk))_{k\in \ZZ^2}$
\cite{Guyon_X_1993_book_cha_armsa}, $R(\lambda,\beta,\aaa)$ can be
estimated via an empirical average $\ests{R}(\lambda,\beta,\aaa)$
computed over $\KK$, provided that the data length
$K=\mathrm{card}(\KK)$ is large enough. Following a procedure
similar to the search implemented for the SUREshrink estimator, we
will subsequently determine optimal values of $\lambda$ and
$\beta$ for this consistent risk estimate. More precisely, the
norms of the ROVs $(\|\bbf{r}(\kk)\|_{\LLa})_{\kk\in \KK}$ are
first sorted in descending order, so that
$\|\bbf{r}(\kk_1)\|_{\LLa} \geq \|\bbf{r}(\kk_2)\|_{\LLa} \geq
\ldots \geq \|\bbf{r}(\kk_{K})\|_{\LLa}$. To study the variations
of $\ests{R}(\lambda,\beta,\aaa)$ w.r.t. $\lambda$, we consider
the case when $\lambda \in I_{i_0}$ with $i_0 \in
\{1,\ldots,K+1\}$ and
\begin{equation}
I_{i_0} =
\begin{cases}
\big[\|\bbf{r}(\kk_{1})\|_{\LLa}^\beta,\infty\big)
& \mbox{if $i_0 = 1$}\\
\big[\|\bbf{r}(\kk_{i_0})\|_{\LLa}^\beta,\|\bbf{r}(\kk_{i_0-1})\|_{\LLa}^\beta\big) & \mbox{if $i_0 \in \{2,\ldots, K\}$}\\
\big[0,\|\bbf{r}(\kk_{K})\|_{\LLa}^\beta\big) & \mbox{if $i_0 = K+1$}.
\end{cases}
\end{equation}
On the interval $I_{i_0}$,
the risk estimate then takes the following form:\footnote{We adopt here the convention
$\sum_{i=1}^0 \cdot = \sum_{i=K+1}^K \cdot = 0$.}
\begin{align}
\ests{R}(\lambda,\beta,\aaa)&=\frac{1}{K}\Big(\sum_{i=1}^{i_0-1}\rho_{\lambda,\beta,\aaa}(\kk_i)+\sum_{i=i_0}^{K}
\rho_{\lambda,\beta,\aaa}(\kk_i)\Big) \\
&=\frac{1}{K}\Big(\lambda^2 \sum_{i=1}^{i_0-1}\alpha_2(\kk_i)
+\lambda \sum_{i=1}^{i_0-1}\alpha_1(\kk_i) +
\sum_{i=1}^{i_0-1}\alpha_0(\kk_i) + \sum_{i=i_0}^{K}
r^2(\kk_i)-(K+1-i_0)\sigma^2 \Big). \label{eq:riskest}
\end{align}
In other words, $\ests{R}(\lambda,\beta,\aaa)$ is a piecewise second-order polynomial function  of $\lambda$.
Assume now that $i_0 \in \{2,\ldots,K\}$. For given values of $\beta$ and $\aaa$, the minimum over $\RR$
of the polynomial in \eqref{eq:riskest} is reached at
\begin{align}
\tilde{\lambda}_{i_0}(\beta,\aaa)=-\frac{\sum_{i=1}^{i_0-1}\alpha_1(\kk_i)}{2\sum_{i=1}^{i_0-1}\alpha_2(\kk_i)}.
\end{align}
The minimum over $\big[\|\bbf{r}(\kk_{i_0})\|_{\LLa}^\beta,\|\bbf{r}(\kk_{i_0-1})\|_{\LLa}^\beta\big]$
of the estimated risk is therefore given by
\begin{equation}
\lambda^*_{i_0}(\beta,\aaa) =
\begin{cases}
\tilde{\lambda}_{i_0}(\beta,\aaa) & \mbox{if $[\|\bbf{r}(\kk_{i_0})\|_{\LLa}^\beta \le \tilde{\lambda}_{i_0}(\beta,\aaa) \le \|\bbf{r}(\kk_{i_0-1})\|_{\LLa}^\beta$}\\
\|\bbf{r}(\kk_{i_0})\|_{\LLa}^\beta &
\mbox{if $\tilde{\lambda}_{i_0}(\beta,\aaa) < \|\bbf{r}(\kk_{i_0})\|_{\LLa}^\beta$}\\
\|\bbf{r}(\kk_{i_0-1})\|_{\LLa}^\beta &
\mbox{if $\tilde{\lambda}_{i_0}(\beta,\aaa) >\|\bbf{r}(\kk_{i_0-1})\|_{\LLa}^\beta$}.
\end{cases}
\end{equation}
The minimizers $\lambda^*_{1}(\beta,\aaa)$ and $\lambda^*_{K+1}(\beta,\aaa)$ of the
estimated risk over $I_1$ and $I_{K+1}$ can be found in a similar way.
The global minimizer $\lambda^*(\beta,\aaa)$ of the estimated risk is subsequently computed as
\begin{equation}
\lambda^*(\beta,\aaa) = \arg\min_{
(\lambda^*_{i_0}(\beta,\aaa))_{1\le i_0 \le K+1}}
\ests{R}(\lambda^*_{i_0}(\beta,\aaa),\beta,\aaa).
\end{equation}

To determine the optimal value $\beta^*(\aaa)$ of the exponent
$\beta$, we can then proceed to an exhaustive search over a
set $\mathcal{V}$ of possible values for this parameter by
choosing
\begin{equation}
\beta^*(\aaa)=\arg\min_{\beta\in \mathcal{V}}
\ests{R}(\lambda^*(\beta,\aaa),\beta,\aaa).
\end{equation}
In our experiments, it was observed that a restricted set of a few 
search values is sufficient to get good results.

\subsection{Iterative optimization algorithm}
The optimal expression of the vector $\aaa$ is
derived in a closed form in Section
\ref{sec:estim_q} as a function of the parameters
$\lambda$ and $\beta$. In this way, the optimization problem simply 
reduces to the determination of the latter two parameters. On the other hand, given $\aaa$,
a procedure for determining the optimal values of $\lambda$ and
$\beta$ is described in Section \ref{sec:estim_beta_lambda}. 
In order to get optimized
values of the estimator parameters, we therefore propose to apply the
following iterative optimization approach:
\begin{enumerate}
\item Initialization: Fix %
$\mathbf{P}$ and $\mathcal{V}$. Set the iteration number $p=1$ and
$\aaa^{(0)} = [1,0,\ldots,0]^\top \in \RR^{d'}$
\item\label{itern} Iteration $p$
\begin{enumerate}
\item Set $\beta^{(p)} = \beta^*(\aaa^{(p-1)})$ and
$\lambda^{(p)} = \lambda^*(\beta^{(p)},\aaa^{(p-1)})$
as described in Section \ref{sec:estim_beta_lambda}.
\item Set $\aaa^{(p)} = \aaa^*(\lambda^{(p)},\beta^{(p)})$
using \eqref{eq:opta} where the expectations are replaced by sample estimates.
\end{enumerate}
\item Set $p\leftarrow p+1$ and goto step \ref{itern}
until convergence.
\item Return the  optimized values $(\lambda^{(p)},\beta^{(p)},\aaa^{(p)})$ of the parameters.
\end{enumerate}
We point out that, although we were not able to prove the convergence of the optimized parameters, the generated sequence 
$(\ests{R}(\lambda^{(p)},\beta^{(p)},\aaa^{(p)}))_p$ is a
decreasing convergent sequence. This means that improved
parameters are generated at each iteration of the algorithm.

\section{Multicomponent  wavelet denoising}
\label{sec:Multi_wav_den} Our objective here is to apply the
nonlinear estimator developed in the previous section  to noise
reduction in degraded multicomponent images by considering
wavelet-based approaches.
The original multichannel image is composed of $B \in
{\mathbb{N}}^*$ components $s^{(b)}$ of size $L \times L$, with
$b\in \{1, \ldots, B\}$. Each image component $s^{(b)}$ is
corrupted by an additive noise $n^{(b)}$, which is assumed
independent of the images of interest. Consequently, we obtain the
following noisy observation  field $r^{(b)}$ defined by:
\begin{equation}
{\forall \mathbf{k}} \in \KK,\qquad r^{(b)}({\mathbf{k}})= s^{(b)}({\mathbf{k}})+ n^{(b)}({\mathbf{k}}),
\label{eq:observ_band}
\end{equation}
where $\KK=\{1,\ldots,L\}^2$.
Following a multivariate approach, we define:
\begin{equation}
{\forall \mathbf{k}} \in \KK,\qquad
\left\{
\begin{array}{lll}
{\mathbf{s}}({\mathbf{k}})&\eqdef&[s^{(1)}({\mathbf{k}}), \ldots,s^{(B)}({\mathbf{k}})]^\top\\
{\mathbf{n}}({\mathbf{k}})&\eqdef&[n^{(1)}({\mathbf{k}}), \ldots,n^{(B)}({\mathbf{k}})]^\top\\
{\mathbf{r}}({\mathbf{k}})&\eqdef&[(r^{(1)}({\mathbf{k}}),
\ldots,r^{(B)}({\mathbf{k}})]^\top
\end{array}
\right..
\end{equation}
Obviously, the observation model (\ref{eq:observ_band}) can be
rewritten as ${\forall \mathbf{k}} \in\KK,\quad
\mathbf{r}({\mathbf{k}})= \mathbf{s}({\mathbf{k}})+
\mathbf{n}({\mathbf{k}}).$ In many optical systems, the noise
stems from a combination of photonic and electronic noises
cumulated with quantization errors. Subsequently, we will assume
that the noise vector process ${\mathbf{n}}$ is zero-mean iid
Gaussian with covariance matrix
${\GGamma}^{(\mathbf{n})}$. In
\cite{Abrams_M_2002_sam_rirpgmrss} and
\cite{Corner_B_2003_ijrs_noi_ersidm}, this was shown to constitute
a realistic assumption for satellite systems. It is worth noticing
that a non diagonal matrix ${\GGamma}^{(\mathbf{n})}$
indicates that inter-component correlations exist between
co-located noise samples.

Hereafter, we will use two decompositions. The first one consists
in a critically decimated $M$-band wavelet transform whereas the
second one, corresponds to an $M$-band dual-tree wavelet
decomposition we recently proposed
\cite{Chaux_C_2006_tip_ima_adtmbwt} which permits a
directional analysis of images.

\subsection{$M$-band wavelet basis estimation}

\subsubsection{Model}\label{se:modelwav}
We first consider an $M$-band orthonormal discrete wavelet
transform (DWT) \cite{Steffen_P_1993_tsp_the_rmbwb} over $J$
resolution levels applied, for each channel $b$, to the
observation field $r^{(b)}$. This decomposition produces $M^2-1$
wavelet subband sequences $r_{j,\mathbf{m}}^{(b)}$, $\mathbf{m}
\in \NMSQ \setminus \{(0,0)\}$, each of size $L_j \times L_j$
(where $L_j\,=\,L/M^j$)\footnote{For simplicity, $L$ is assumed to
be divisible by $M^J$.}, at every resolution level $j$ and an
additional approximation sequence $r_{J,\mathbf{0}}^{(b)}$ of size
$L_J\times L_J$, at resolution level $J$. %

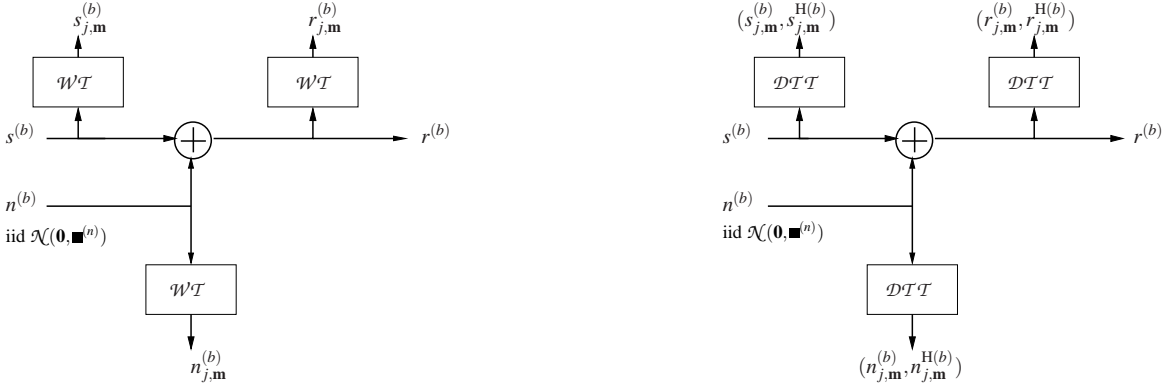
\begin{figure}[ht]
\begin{center}
\begin{tabular}{lcr}
\input{model_wavelets.pstex_t} & \hspace{2cm} & \input{model_dualtree.pstex_t}\\
\end{tabular}
\caption{Considered models in the wavelet transform domain (left)
and in the dual-tree transform domain (right). \label{fig:model_transf}}
\end{center}
\vspace*{-0.8cm} \end{figure}

On the one hand, the linearity of the DWT yields (see. Fig.
\ref{fig:model_transf}): ${\forall \mathbf{k}} \in \KK_j,\quad
{\mathbf{r}}_{j,\mathbf{m}}({\mathbf{k}})={\mathbf{s}}_{j,\mathbf{m}}({\mathbf{k}})+{\mathbf{n}}_{j,\mathbf{m}}({\mathbf{k}})$
where $\KK_j = \{1,\ldots,L_j\}^2$ and
\begin{align*}
&{\mathbf{s}}_{j,\mathbf{m}}({\mathbf{k}})\,\eqdef\,[s_{j,\mathbf{m}}^{(1)}({\mathbf{k}}),\ldots,s_{j,\mathbf{m}}^{(B)}({\mathbf{k}})]^\top,\\
&{\mathbf{r}}_{j,\mathbf{m}}({\mathbf{k}})\,\eqdef\,[r_{j,\mathbf{m}}^{(1)}({\mathbf{k}}),\ldots,r_{j,\mathbf{m}}^{(B)}({\mathbf{k}})]^\top,\\
&{\mathbf{n}}_{j,\mathbf{m}}({\mathbf{k}})\,\eqdef\,[n_{j,\mathbf{m}}^{(1)}({\mathbf{k}}),\ldots,n_{j,\mathbf{m}}^{(B)}({\mathbf{k}})]^\top.
\end{align*}
On the other hand, the orthonormality of the DWT preserves the
spatial whiteness of ${\mathbf{n}}_{j,\mathbf{m}}$. More
specifically, it is easily shown that the latter field is an
i.i.d.
${\mathcal{N}}({\mathbf{0}},{\GGamma}^{(\mathbf{n})})$ random vector process.\\
A final required assumption is that the random vectors
$({\mathbf{s}}_{j,\mathbf{m}}({\mathbf{k}}))_{k\in\KK}$ are
identically distributed for any given value of $(j,\mathbf{m})$.

\subsubsection{Associated estimator}
As described in Section \ref{sec:prop_estim}, our estimator can be
directly applied to the $M$-band DWT coefficients.
As in conventional approaches, the approximation
coefficients (\emph{i.e.} $j=J$ and $\mm=(0,0)$) are kept
untouched. The parameters $\lambda_{j,\mathbf{m}}$,
$\beta_{j,\mathbf{m}}$ and $\mathbf{q}_{j,\mathbf{m}}$ can be
determined adaptively, for every subband $(j,\mathbf{m})$ and
every component $b$. In this case, the ROV can be scalar, spatial,
inter-component or combined spatial/inter-component. More detailed examples will be given in Section \ref{sec:simuls}.

\subsection{$M$-band dual-tree wavelet frame estimation}

\subsubsection{A brief overview of the decomposition}
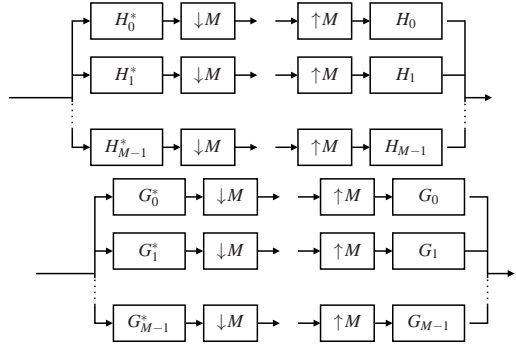
\begin{figure}[ht]
\begin{center}
\input{fb_M_ortho_dual.pstex_t}
\caption{Pair of analysis/synthesis $M$-band para-unitary filter
banks. \label{fbM}}
\end{center}
\vspace*{-0.8cm} \end{figure}

\begin{figure}[ht]
\begin{center}
\input{Mbanddualtree2D.pstex_t}
\caption{Dual-tree $2D$. \label{DTT2D}}
\end{center}
\vspace*{-0.8cm} \end{figure}
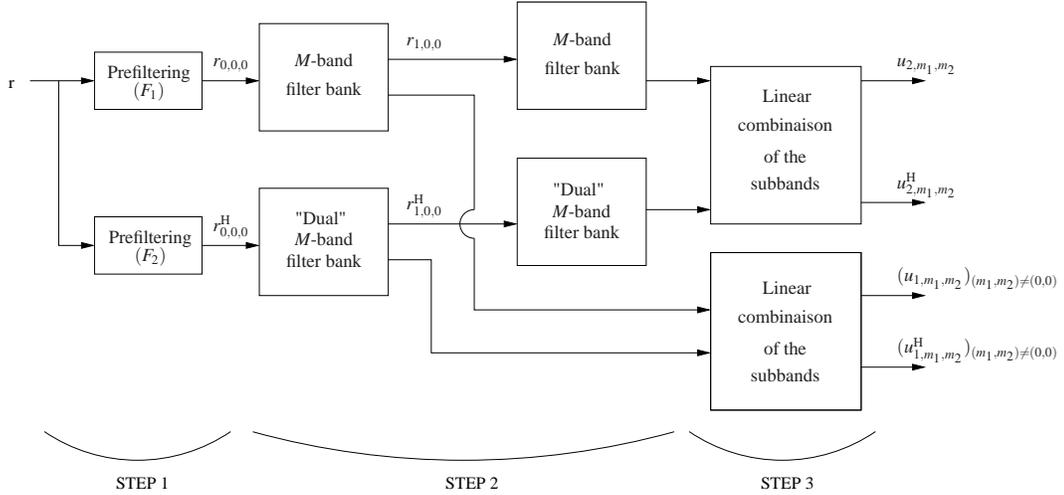

The $M$-band real dual-tree transform (DTT) consists in performing
two separable $M$-band orthonormal wavelet decompositions in
parallel as illustrated by Fig. \ref{fbM}. The one-dimensional
wavelets $(\psi_m)_{m \in \NMs}$ corresponding to the primal tree
(upper branch) are assumed known and  the ``dual tree''
ones $(\psi^\HH_m)_{m\in \NMs}$ (used in the lower branch) are
built so that they define  Hilbert pairs with the primal ones.
This reads in the frequency domain: $\forall m \in \NMs,\;
\widehat{\psi}^\HH_m(\omega) = -\imath \, \textrm{sign}(\omega)
\widehat{\psi}_m(\omega)$. Details of construction are given in
\cite{Chaux_C_2006_tip_ima_adtmbwt} and the global scheme of the
decomposition is shown in Fig.~\ref{DTT2D}. An important point is
that the dual-tree decomposition includes a post-processing,
consisting of a linear isometric combination of the primal/dual
subbands (see Fig. \ref{DTT2D}). This post-processing constitutes
an essential step for obtaining a directional analysis.
Finally, two sets of coefficients (primal and dual ones) are
obtained, %
which means that this representation involves a limited redundancy
of a factor two.

\subsubsection{Model}
Applying this decomposition to a multichannel image having $B$ components
and using similar notations to Section~\ref{se:modelwav},
we obtain the following coefficients for the original data,
the observed ones and the noise, respectively:
\begin{itemize}
\item before post-processing:
$({\mathbf{s}}_{j,\mathbf{m}}({\mathbf{k}}),{
\mathbf{s}}^\HH_{j,\mathbf{m}}({\mathbf{k}})),
({\mathbf{r}}_{j,\mathbf{m}}({\mathbf{k}}),{\mathbf{r}}^\HH_{j,\mathbf{m}}({\mathbf{k}})),
({\mathbf{n}}_{j,\mathbf{m}}({\mathbf{k}}),{\mathbf{n}}^\HH_{j,\mathbf{m}}({\mathbf{k}}))$;
\item after post-processing:
$({\mathbf{v}}_{j,\mathbf{m}}({\mathbf{k}}),{\mathbf{v}}^\HH_{j,\mathbf{m}}({\mathbf{k}})),
({\mathbf{u}}_{j,\mathbf{m}}({\mathbf{k}}),{\mathbf{u}}^\HH_{j,\mathbf{m}}({\mathbf{k}})),
({\mathbf{w}}_{j,\mathbf{m}}({\mathbf{k}}),{\mathbf{w}}^\HH_{j,\mathbf{m}}({\mathbf{k}}))$.
\end{itemize}
Note that a post-processing is not applied to all subbands (see
\cite{Chaux_C_2006_tip_ima_adtmbwt}) as the Hilbert condition is only
verified by mother wavelets. As a consequence, the linear
isometric combination is not performed for subbands processed by
low pass filters. More precisely, the post-processing consists of
the following unitary transform of the detail coefficients: for
all $\mm \in \NMsSQ$,
\begin{equation}
\forall \kk \in \KK_j,\qquad \mathbf{w}_{j,\mm}(\kk) =
\frac{1}{\sqrt{2}}\big(\mathbf{n}_{j,\mm}(\kk)+\mathbf{n}_{j,\mm}^\HH(\kk)\big)
\quad \mbox{and}
\quad \mathbf{w}_{j,\mm}^\HH(\kk) = \frac{1}{\sqrt{2}}
\big(\mathbf{n}_{j,\mm}(\kk)-\mathbf{n}_{j,\mm}^\HH(\kk)\big).
\label{eq:transfnoise2}
\end{equation}
Similar relations hold for the original and observed data.
Furthermore, invoking the linearity property of the transform,
these coefficients are related by (see. Fig.
\ref{fig:model_transf} (right)):
\begin{align}
{\forall \mathbf{k}} \in \KK_j,\quad
&{\mathbf{r}}_{j,\mathbf{m}}({\mathbf{k}})={\mathbf{s}}_{j,\mathbf{m}}({\mathbf{k}})+{\mathbf{n}}_{j,\mathbf{m}}({\mathbf{k}}) \quad \mbox{and} \quad {\mathbf{u}}_{j,\mathbf{m}}({\mathbf{k}})={\mathbf{v}}_{j,\mathbf{m}}({\mathbf{k}})+{\mathbf{w}}_{j,\mathbf{m}}({\mathbf{k}}) \nonumber\\
&{\mathbf{r}}^\HH_{j,\mathbf{m}}({\mathbf{k}})={\mathbf{s}}^\HH_{j,\mathbf{m}}({\mathbf{k}})+{\mathbf{n}}^\HH_{j,\mathbf{m}}({\mathbf{k}})
\quad \quad \quad \; \;
{\mathbf{u}}^\HH_{j,\mathbf{m}}({\mathbf{k}})={\mathbf{v}}^\HH_{j,\mathbf{m}}({\mathbf{k}})+{\mathbf{w}}^\HH_{j,\mathbf{m}}({\mathbf{k}}).
\end{align}

\subsubsection{Noise statistical properties}
In our recent work
\cite{Chaux_C_2005_gretsi_etu_bambad,Chaux_C_tit_xxxx_tit_noi_cpdtwd},
a detailed analysis of the noise statistical properties after such
a dual tree decomposition has been performed. In the sequel, some
of the main results we obtained are briefly summarized. Let us
recall the definition of the deterministic cross-correlation
function between the primal and dual wavelets: for all $(m,m') \in
\NN_M^2$,
\begin{equation}
\forall \tau \in \RR,\qquad \gamma_{m,m'}(\tau)   =  \int_{-\infty}^\infty
\psi_{m}(x)\psi_{m'}^\HH(x-\tau)\;dx. \label{eq:intnoise}
\end{equation}
We have obtained the following
expressions for the covariance fields: for all $j\in\ZZ$, $\mm = (m_1,m_2) \in \NMSQ$,\linebreak $\mm' = (m_1',m_2') \in \NMSQ$, $\kk = (k_1,k_2)\in \KK_j$ and $\kk'=(k_1',k_2')\in \KK_j$,
\begin{gather*}
   \left.\begin{aligned}
      \E[\mathbf{n}_{j,\mm}(\kk)(\mathbf{n}_{j,\mm'}(\kk'))^\top]   \\
\E[\mathbf{n}_{j,\mm}^\HH(\kk)
(\mathbf{n}_{j,\mm'}^\HH(\kk'))^\top]
\end{aligned} \right\}
= {\GGamma}^{(\mathbf{n})} \delta_{m_1-m_1'} \delta_{m_2-m_2'} \delta_{k_1-k_1'} \delta_{k_2-k_2'} \\
\E[\mathbf{n}_{j,\mm}(\kk)(\mathbf{n}_{j,\mm'}^\HH(\kk'))^\top] =
{\GGamma}^{(\mathbf{n})}\gamma_{m_1,m_1'}(k_1'-k_1)
\gamma_{m_2,m_2'}(k_2'-k_2)\,.
\end{gather*}
It can be further noticed that, for $\mm\neq {\mathbf{0}}$, the random
vectors
$\mathbf{n}_{j,\mm}(\kk)$ and $\mathbf{n}_{j,\mm}^\HH(\kk)$
at a given location $\kk$
are mutually uncorrelated. \\
After post-processing, the covariances of the transformed noise
coefficient fields  can be easily deduced from %
\eqref{eq:transfnoise2}: for all $(\mm,\mm') \in \NMsSQ$ and
$(\kk,\kk')\in\KK_j^2$,
\begin{align}
\E[\mathbf{w}_{j,\mm}(\kk)(\mathbf{w}_{j,\mm'}(\kk'))^\top]=&
\E[\mathbf{n}_{j,\mm}(\kk)(\mathbf{n}_{j,\mm'}(\kk'))^\top]+
\E[\mathbf{n}_{j,\mm}(\kk)(\mathbf{n}_{j,\mm'}^{\HH}(\kk'))^\top]
\label{eq:2Drot1}\\
\E[\mathbf{w}_{j,\mm}^{\HH}(\kk)(\mathbf{w}_{j,\mm'}^{\HH}(\kk'))^\top]=&
\E[\mathbf{n}_{j,\mm}(\kk)(\mathbf{n}_{j,\mm'}(\kk'))^\top]-
\E\{\mathbf{n}_{j,\mm}(\kk)(\mathbf{n}_{j,\mm'}^{\HH}(\kk'))^\top]
\label{eq:2Drot2}\\
\E[\mathbf{w}_{j,\mm}(\kk)(\mathbf{w}_{j,\mm'}^{\HH}(\kk'))^\top]=&\mathbf{0}.
\label{eq:2Drot3}
\end{align}
In summary, noise coefficients are inter-tree correlated before
the post-transform whereas after the post-transform, they are
spatially correlated. This constitutes an important consequence of
the post-processing stage.

\subsubsection{Associated estimator}
In the $M$-band DTT case, the primal and dual coefficients are both estimated. For each
component $b \in  \{1,\ldots,B\}$, the estimator reads: for the subbands which are not linearly combined ($\mathbf{m}\not\in\NMs$),
\begin{align}
&\estw{s}{(b)}_{j,\mathbf{m}}(\kk) =
\eta_{\lambda_{j,\mathbf{m}}^{(b)}}(\|\bbf{r}_{j,\mathbf{m}}^{(b)}(\kk)\|_{\LLa}^{\beta_{j,\mathbf{m}}^{(b)}})
\;
(\mathbf{q}_{j,\mathbf{m}}^{(b)})^{\top}\bbf{r}_{j,\mathbf{m}}^{(b)}(\kk)\\
&\estwl{s}{\HH(b)}_{j,\mathbf{m}}(\kk) =
\eta_{\lambda_{j,\mathbf{m}}^{\HH(b)}}(\|(\bbf{r}_{j,\mathbf{m}}^{\HH(b)}(\kk))\|_{\LLa}^{\beta_{j,\mathbf{m}}^{\HH(b)}})
\;
\big(\mathbf{q}_{j,\mathbf{m}}^{\HH(b)}\big)^{\top}\bbf{r}_{j,\mathbf{m}}^{\HH(b)}(\kk), \label{estimatorDTT}
\end{align}
and, for the combined subbands ($\mathbf{m}\in\NMs$),
\begin{align}
&\estw{v}{(b)}_{j,\mathbf{m}}(\kk) =
\eta_{\lambda_{j,\mathbf{m}}^{(b)}}(\|\bbf{u}_{j,\mathbf{m}}^{(b)}(\kk)\|_{\LLa}^{\beta_{j,\mathbf{m}}^{(b)}})
\;
(\mathbf{q}_{j,\mathbf{m}}^{(b)})^{\top}\bbf{u}_{j,\mathbf{m}}^{(b)}(\kk)\\
&\estwl{v}{\HH(b)}_{j,\mathbf{m}}(\kk) =
\eta_{\lambda_{j,\mathbf{m}}^{\HH(b)}}(\|(\bbf{u}_{j,\mathbf{m}}^{\HH(b)}(\kk))\|_{\LLa}^{\beta_{j,\mathbf{m}}^{\HH(b)}})
\;
\big(\mathbf{q}_{j,\mathbf{m}}^{\HH(b)}\big)^{\top}\bbf{u}_{j,\mathbf{m}}^{\HH(b)}(\kk),
\label{estimatorDTTH}
\end{align}
where $\bbf{r}_{j,\mathbf{m}}^{\HH(b)}(\kk)$ and
$\bbf{r}_{j,\mathbf{m}}^{\HH(b)}(\kk)$ (resp.
$\bbf{u}_{j,\mathbf{m}}^{(b)}(\kk)$ and
$\bbf{u}_{j,\mathbf{m}}^{\HH(b)}(\kk)$) are the ROVs for the
primal and dual coefficients before (resp. after)
post-transformation. Similarly to the DWT case,
$(\lambda_{j,\mathbf{m}},\beta_{j,\mathbf{m}},\mathbf{q}_{j,\mathbf{m}})$
and
$(\lambda_{j,\mathbf{m}}^\HH,\beta_{j,\mathbf{m}}^\HH,\mathbf{q}_{j,\mathbf{m}}^\HH)$
can be adaptively determined by minimizing the quadratic risk over
the frame coefficients for every subband $(j,\mathbf{m})$ and
every component $b$ in each tree. Furthermore, the approximation
coefficients are also kept untouched. The denoised multichannel
images are then obtained from the estimated wavelet coefficients
by inverting the DTT using the optimal reconstruction developed in
\cite{Chaux_C_2006_tip_ima_adtmbwt}. In this case, a great
flexibility exists in the choice of the ROV since the latter can
be scalar, spatial, inter-component, inter-tree or combined
spatial/inter-component/inter-tree as will be illustrated in the
next section.

\section{Numerical results}
\label{sec:simuls} We now provide numerical examples showing the
efficiency of the proposed method. In our simulations, we consider
different multichannel remote sensing images. For the sake of
clarity, we only provide experimental results concerning two
multispectral images. The first one designated as Tunis
corresponds  to a part of a SPOT3 scene depicting a urban area of
the city of Tunis ($B=3$). The second one named Trento is a
Landsat Thematic Mapper image having initially seven channels. The
thermal component (the sixth component) has been discarded since
it is not similar to the remaining ones. Hence, the test image
Trento  is a $B=6$ component image.  In order to obtain reliable
results from a statistical viewpoint, Monte Carlo simulations have
been conducted. According to our experiments, averaging the mean
square error over five noise realizations is sufficient to obtain
consistent quantitative evaluations.

In the following, we discuss several topics: in particular, we
compare our method with other recently proposed estimators,
possibly having a multivariate structure. Then, we consider
different pre-processings that can be performed on the
multichannel data before applying the estimator, thus expecting
improved results. The ROV being defined in a generic way in the
previous section, we also study the influence of specific choices
of this ROV on the denoising performance as well as the influence
of the wavelet choice (considering various $M$-band filter banks).
When different decompositions are performed, we set the maximum
decomposition level so that the size of the approximation fields
remain the same. Consequently, we decompose the images over $2$
levels for a $4$-band filter bank structure and $4$ levels for a
dyadic one.

If $\sigma^{(b)}$
denotes the standard deviation of the clean multichannel component
$s^{(b)}$ (of size $L_1 \times L_2$) we define the initial and the final signal to noise
ratios $\textrm{SNR}_\textrm{initial}^{(b)}$ and,
$\textrm{SNR}^{(b)}_\textrm{final}$ in the $b$-th channel as:
\begin{equation}
\textrm{SNR}^{(b)}_\textrm{initial}\,\eqdef\,10\log_{10}\left(\frac{(\sigma^{(b)})^2L_1 L_2}{\|s^{(b)}-r^{(b)}\|^2}\right),
\quad \mbox{and} \quad
\textrm{SNR}^{(b)}_\textrm{final}\,\eqdef\,10\log_{10}\left(\frac{(\sigma^{(b)})^2L_1 L_2}{\|s^{(b)}-\hat{s}^{(b)}\|^2}\right).
\end{equation}
Then, all the $B$ channel contributions are averaged into
global values of the initial and final signal to noise ratio
$\textrm{SNR}_\textrm{initial}$ and,
$\textrm{SNR}_\textrm{final}$.

\subsection{Comparison with existing methods}
\label{se:compare}
We aim in this section at comparing the proposed approach
with several
existing denoising methods which
are briefly described in Table \ref{tab:exis_method}.
\begin{table}
\begin{center}
\caption{Brief description of the tested
methods.\label{tab:exis_method}}
\begin{tabular}{|l|l|l||l|l|l|}
\hline
Acronym & Description & Ref. &  Acronym & Description & Ref.\\
\hline\hline
Biv. & Bivariate shrinkage method & \cite{Sendur_L_2002_spl_biv_slve} & \multicolumn{3}{c|}{Multivariate methods}\\
\hline
BLS-GSM & Bayesian Least Squares (BLS)  &  \cite{Portilla_J_2003_tip_ima_dsmgwd} & ProbShrink  & Multivariate method for $3$-band images using & \cite{Pizurica_A_2006_tip_est_ppsimsmid}\\
 & Gaussian Scale Mixture (GSM) &  & ($ . \times . $) &  critically decimated DWT  and taking into & \\
 & using critically decimated DWT & & & account a ($ . \times . $) neighborhood in each channel & \\
\hline
BLS-GSM & BLS-GSM using critically  &  \cite{Portilla_J_2003_tip_ima_dsmgwd} & ProbShrink  & Multivariate method for $3$-band images & \cite{Pizurica_A_2006_tip_est_ppsimsmid}\\
 + parent & decimated DWT and taking into  &  & red. ($ . \times . $) & using undecimated DWT and taking into & \\
 & account the parent coefficient &  & & account a ($ . \times . $) neighborhood in each channel & \\
\hline
BLS-GSM& BLS-GSM using a full&  \cite{Portilla_J_2003_tip_ima_dsmgwd} & Surevect & Estimator based on an extended SURE & \cite{Benazza05}\\
  red. &  steerable pyramid & & & approach using a critically decimated DWT & \\
&  (redundant transform) &  & & & \\
\hline
Curvelets & Block estimator using curvelet & \cite{Candes_E_2006_siam-mms_fas_dct}\\
 & transform: $7.5$ times redundant & \\
\cline{1-3}
\end{tabular}
\end{center}\vspace*{-0.8cm}
\end{table}
Tests are
performed on a $512 \times 512$ SPOT$3$ image of Tunis city
($B=3$) (as some multivariate methods are limited to
$3$-band images) corrupted by an additive zero-mean white Gaussian noise with
covariance matrix $\GGamma^{(\mathbf{n})}_1=\sigma^2\;
\mathbf{I}_B$, where $\mathbf{I}_B$ denotes the
identity matrix of size $B \times B$.

We first study techniques that use orthogonal wavelet transforms. We employ Daubechies wavelets of order $4$
in all the following estimators:
\begin{enumerate}
\item the Bivariate shrinkage, which takes into account
inter-scale dependencies, the last level being processed by
inverting children and parent role
\cite{Sendur_L_2002_spl_biv_slve}; \item the BLS-GSM method
developed in \cite{Portilla_J_2003_tip_ima_dsmgwd} including or
not the parent neighborhood and considering a $3 \times 3$ spatial
neighborhood;\footnote{We use the toolbox available from
Portilla's website
\texttt{http://www.io.csic.es/PagsPers/JPortilla/}.} \item the
ProbShrink estimator \cite{Pizurica_A_2006_tip_est_ppsimsmid} for
multivariate data with a $3\times 3$ spatial neighborhood (in each
channel);\footnote{We use the toolbox available from
Pi\u{z}urica's website
\texttt{http://telin.rug.ac.be/$\sim$sanja/}.} 
\item the Surevect
estimator \cite{Benazza05}, which only takes into account
multicomponent statistical dependencies;
\item the proposed
estimator where the set of values taken by
$\beta_{j,\mathbf{m}}^{(b)}$ is $\mathcal{V} = \{0.5,1,1.5,2\}$,
the ROV is represented in Fig. \ref{fig:ROVDWT}(b).
A subspace constraint is added on the vector
$\mathbf{q}_{j,\mathbf{m}}^{(b)}$ so that
$(\mathbf{q}_{j,\mathbf{m}}^{(b)})^\top
\bbf{r}_{j,\mathbf{m}}^{(b)}(\kk)$ reduces to a linear combination
of the multichannel data at the considered location and the 4 spatial nearest neighbors.
\end{enumerate}
\begin{table}[htbp]
\begin{center}
\caption{Denoising results (average values computed over $3$
channels) on Tunis image using non redundant orthogonal transforms
(see Tab. \ref{tab:exis_method}) with Daubechies wavelets of order
$4$ (length $8$).\label{tab:compDWT}}
\begin{tabular}{| c | c || c | c | c | c | c | c |}
\hline
$\sigma^2$ & $\textrm{SNR}_{\textrm{init}}$ & Biv & Probshrink  & BLS-GSM & BLS-GSM   & Surevect & Proposed\\
&  & & ($3\times3$) & & + parent & DWT & DWT\\ \hline \hline
650.3 & 5.081 & 11.85  & 11.86& 12.05 & 12.14 & 13.08& \textbf{13.42}\\ %
\hline %
410.3 & 7.081& 12.89 & 12.84&13.11 & 13.21  & 14.12 & \textbf{14.52}\\ %
\hline %
258.9 & 9.081 & 13.99& 13.91& 14.26 & 14.36   & 15.24 & \textbf{15.70}\\ %
\hline %
163.3 & 11.08& 15.19 & 15.08& 15.49& 15.60   & 16.43 & \textbf{16.96}\\ %
\hline %
103.1 & 13.08 & 16.49 & 16.37& 16.81& 16.93   & 17.70& \textbf{18.28}\\ %
\hline %
65.03 & 15.08& 17.88& 17.54&18.22 &  18.35   & 19.04 & \textbf{19.65}\\ %
\hline
\end{tabular}
\end{center}
\vspace*{-0.8cm}
\end{table}
The obtained results are provided in Table \ref{tab:compDWT} (the
initial SNRs may be different in each channel although  the noise
variance is fixed). 
For the first three methods, denoising has been performed for each
component of the multichannel data.
For orthogonal wavelets, ProbShrink leads to
better results when it is associated to a spatial neighborhood
than when considering only the pixel value to be estimated.
It performs quite similarly to the Bivariate shrinkage. The
BLS-GSM estimator outperforms these two methods providing a gain
of approximatively $0.2$ dB (up to $0.3$ dB by including the
parent coefficient in the neighborhood). Nevertheless, the
Surevect estimator brings more significant improvements and it can
be observed that our method leads to even better numerical results
whatever the initial noise level is. The new structure of the
estimator coupled with a spatial and spectral block processing may
explain such an improvement. Furthermore, the gain increases as
the initial SNR increases, which is interesting in satellite
imaging where the noise is often of low intensity.
To be fair, we would like to mention that, although Bivariate shrinkage,
Probshrink and BLS-GSM were designed for monochannel image denoising, 
extensions of these methods to the multivariate case could probably be envisaged.

In the monochannel case, it has been reported that the use of redundant transforms often
brings noticeable improvements in denoising~\cite{Candes_E_2006_siam-mms_fas_dct}. We
subsequently compare methods that have been proved to be very efficient when combined with a redundant analysis:
\begin{enumerate}
\item the curvelet denoising~\cite{Candes_E_2006_siam-mms_fas_dct}
using a curvelet frame with a redundancy approximatively equal to
$7.5$ and  a block thresholding;\footnote{We employ the
\texttt{Curvelab 2.0} toolbox which can be downloaded from
\texttt{http://www.curvelet.org}.} 
\item the BLS-GSM method using
steerable pyramids with $8$ orientations, including the parent
neighborhood and a $3 \times 3$ spatial neighborhood as described
in \cite{Portilla_J_2003_tip_ima_dsmgwd}, 
\item the ProbShrink
estimator for multivariate data using undecimated wavelet
transform \cite{Pizurica_A_2006_tip_est_ppsimsmid} (with
Daubechies wavelets of length $8$) and taking into account a
$3\times 3$ or no spatial neighborhood; 
\item the Surevect
estimator \cite{Benazza05}, extended to DTT (with Daubechies wavelets of length $8$);
\item the proposed
estimator using a DTT
where $\mathcal{V} = \{0.5,1,1.5,2\}$, the ROV is represented in Fig.
\ref{fig:ROVDTT}(b).
The vector 
$\mathbf{q}_{j,\mathbf{m}}^{(b)}$
(resp. $\mathbf{q}_{j,\mathbf{m}}^{\HH(b)}$)
is such that
it introduces a linear combination of
the multichannel data in the primal (resp. dual) tree at the considered location and the 4 spatial nearest neighbors.
\end{enumerate}
\begin{table}[htbp]
\begin{center}
\caption{Denoising results (average values computed over $3$
channels) on Tunis image using redundant transforms (see Tab.
\ref{tab:exis_method}) with Daubechies wavelets of order $4$
(length $8$). \label{tab:compTred}}
\begin{tabular}{| c | c || c | c | c | c | c | c |}
\hline
$\sigma^2$ & $\textrm{SNR}_{\textrm{init}}$ & Curvelets &  BLS-GSM red  & Probshrink red & Probshrink red  &  Surevect & Proposed\\ %
& & &+ parent & ($3 \times 3$) & ($1 \times 1$) & DTT & DTT \\
\hline \hline
650.3 & 5.081 &11.91  & 12.92 & 13.00 &13.33& 13.54 & \textbf{13.72}\\ %
\hline %
410.3 & 7.081&12.94  & 14.00 & 14.04 &14.38& 14.59 & \textbf{14.80} \\ %
\hline %
258.9 & 9.081 &14.04  & 15.15 & 15.13 &15.50& 15.70 & \textbf{15.97} \\ %
\hline %
163.3 & 11.081&15.17  & 16.38 &  16.28 &16.68& 16.87 & \textbf{17.21} \\ %
\hline %
103.1 & 13.081 &16.33  & 17.68 & 17.51 &17.92& 18.11 & \textbf{18.52} \\ %
\hline %
65.03 & 15.081&17.56  & 19.04 & 18.76 &19.20& 19.42 & \textbf{19.88} \\ %
\hline
\end{tabular}
\end{center}
\vspace*{-0.8cm}
\end{table}
It is worth pointing out that the same noisy images as used in the non redundant case have been
processed by the redundant transforms. As shown in Table
\ref{tab:compTred}, curvelets do not seem really appropriate in
this multichannel context in spite of their promising results in
the monochannel one. ProbShrink and BLS-GSM methods are
very efficient in the redundant case and ProbShrink shows its
superiority when using an inter-component neighborhood.
The methods using a DTT outperform the existing ones in all the
cases.
We point out that the DTT has a limited redundancy of a factor 2
compared with the other considered redundant decompositions. It
can be noticed that our method provides better results than Surevect.
The observed gain
increases as the initial SNR increases and we obtain
significant improvements with respect to critically decimated
transforms of about $0.25$~dB. It is also interesting to note that
the observed gain in terms of SNR leads to quite visible
differences. In Fig. \ref{fig:zoom1}, cropped versions of the
first channel of the Tunis image are displayed, for a low value of
the initial SNR ($4.66$ dB).
\begin{figure}
\begin{center}
\begin{tabular}{ccc}
\includegraphics[width=4cm]{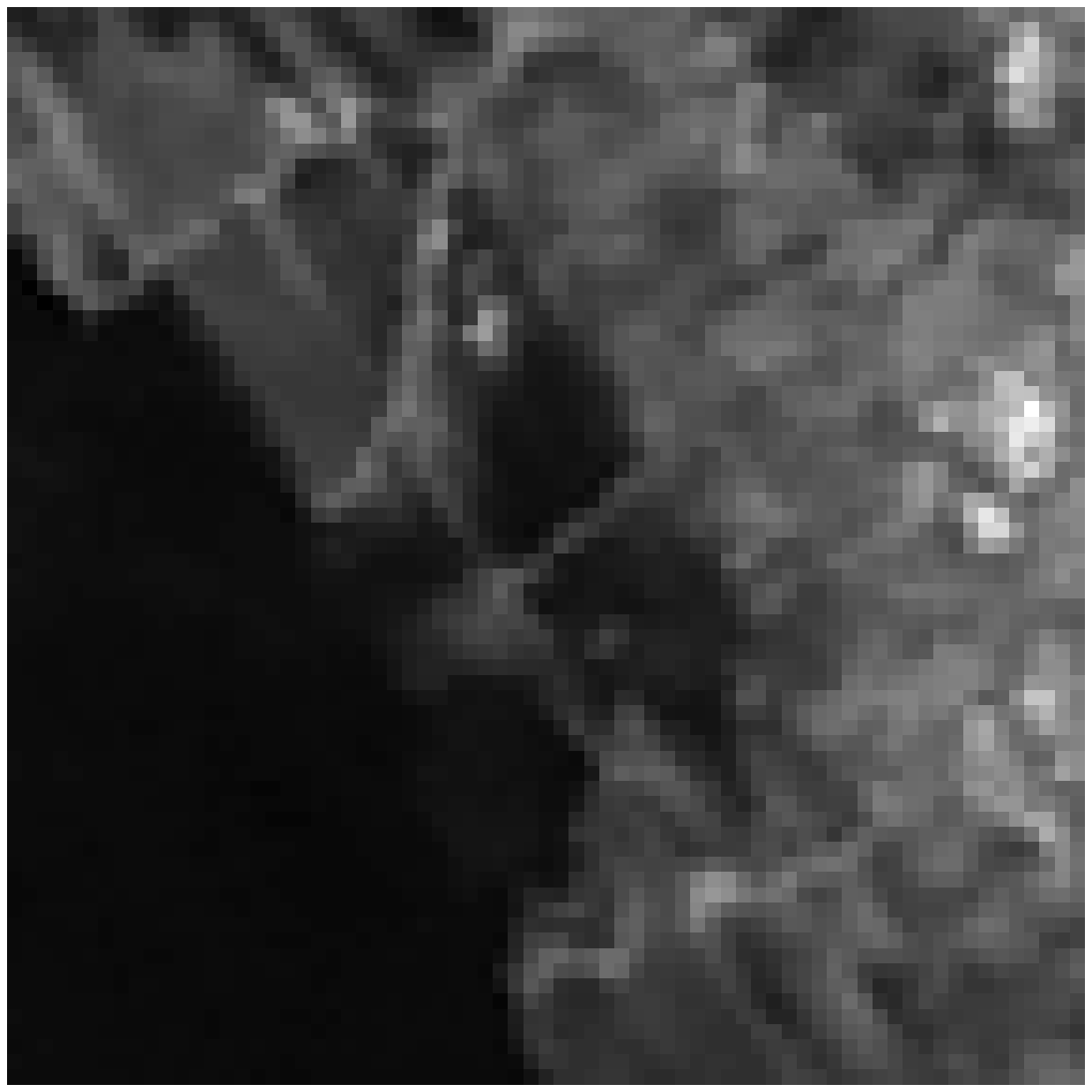} &
\includegraphics[width=4cm]{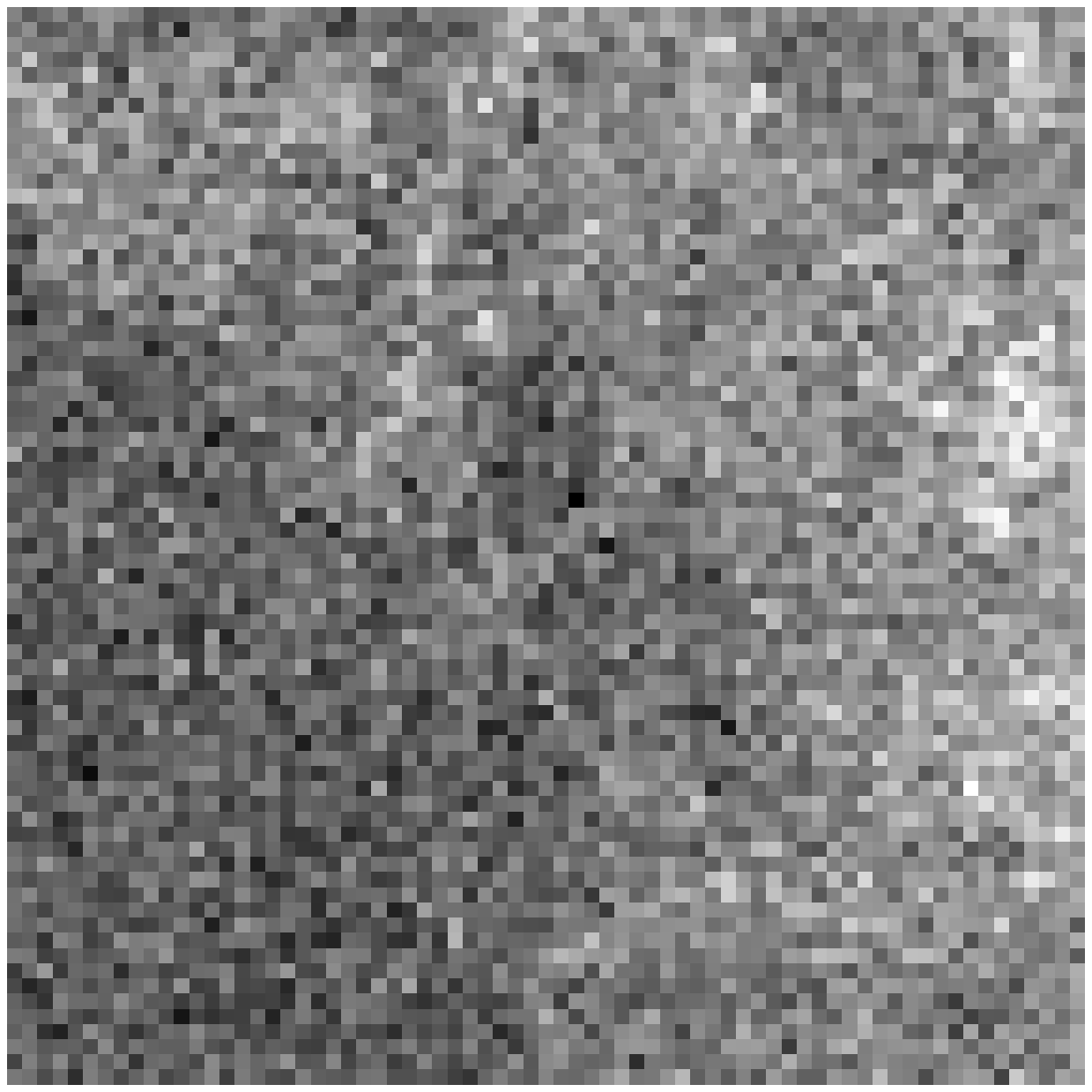} &\includegraphics[width=4cm]{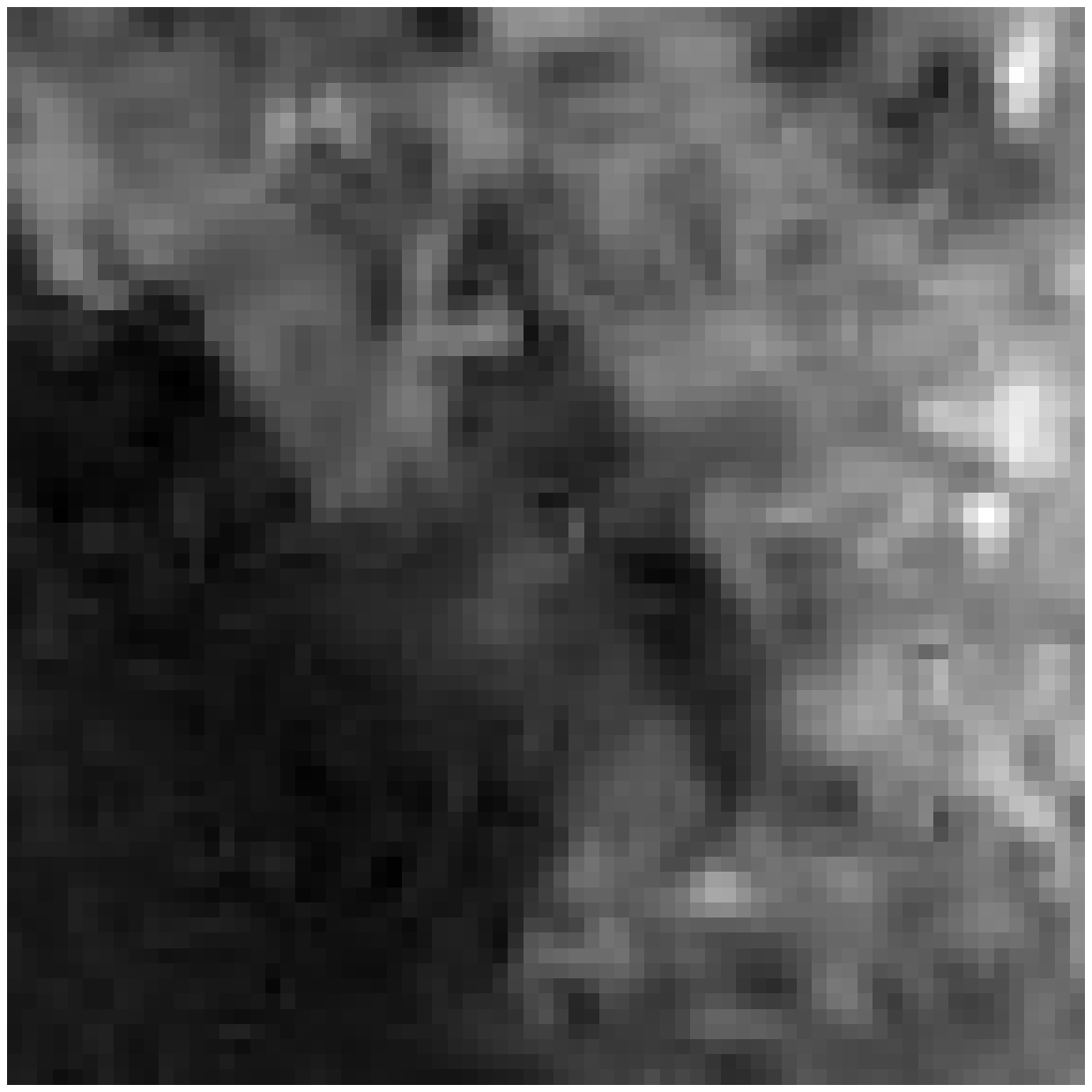}\\
(a) & (b) & (c)\\
\includegraphics[width=4cm]{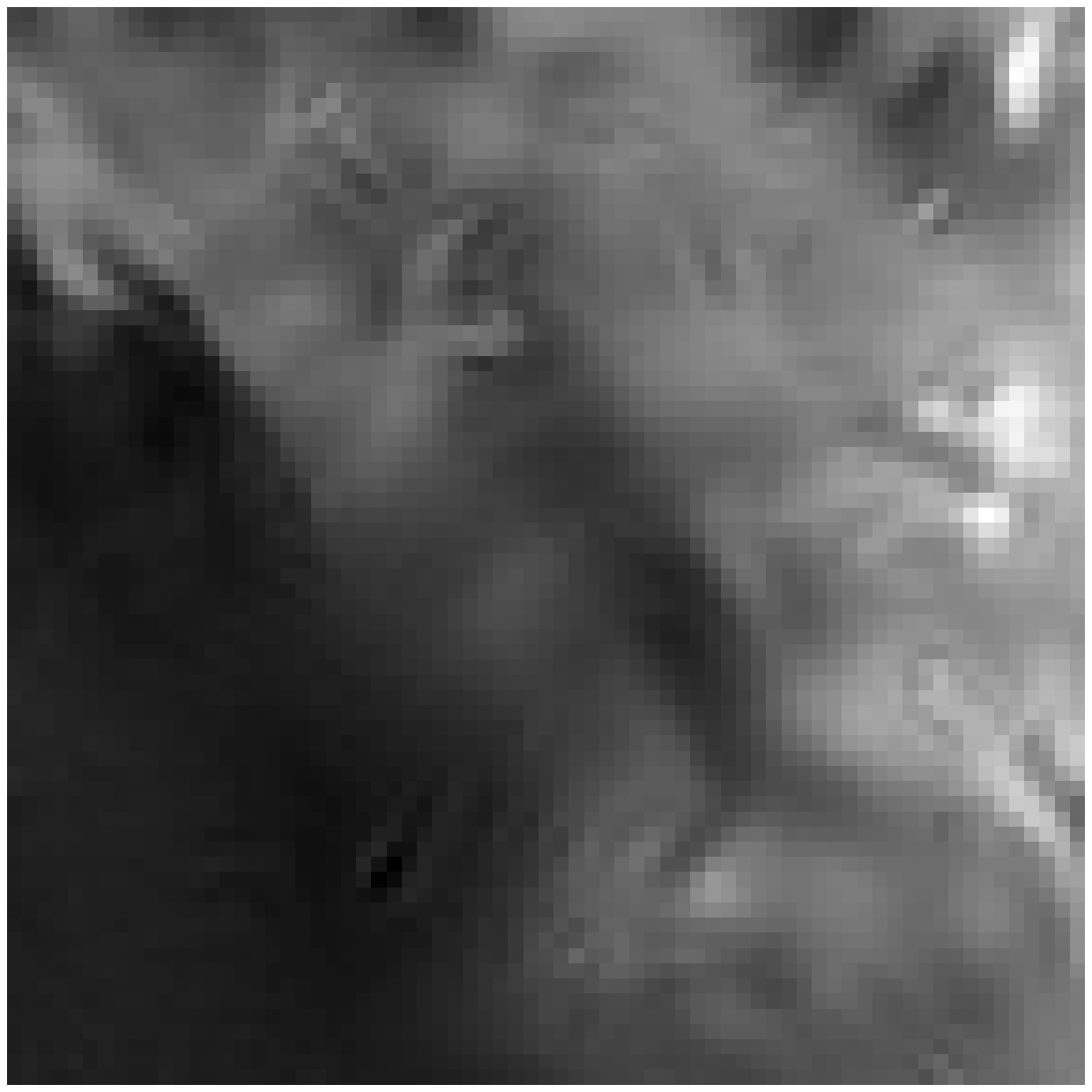} & \includegraphics[width=4cm]{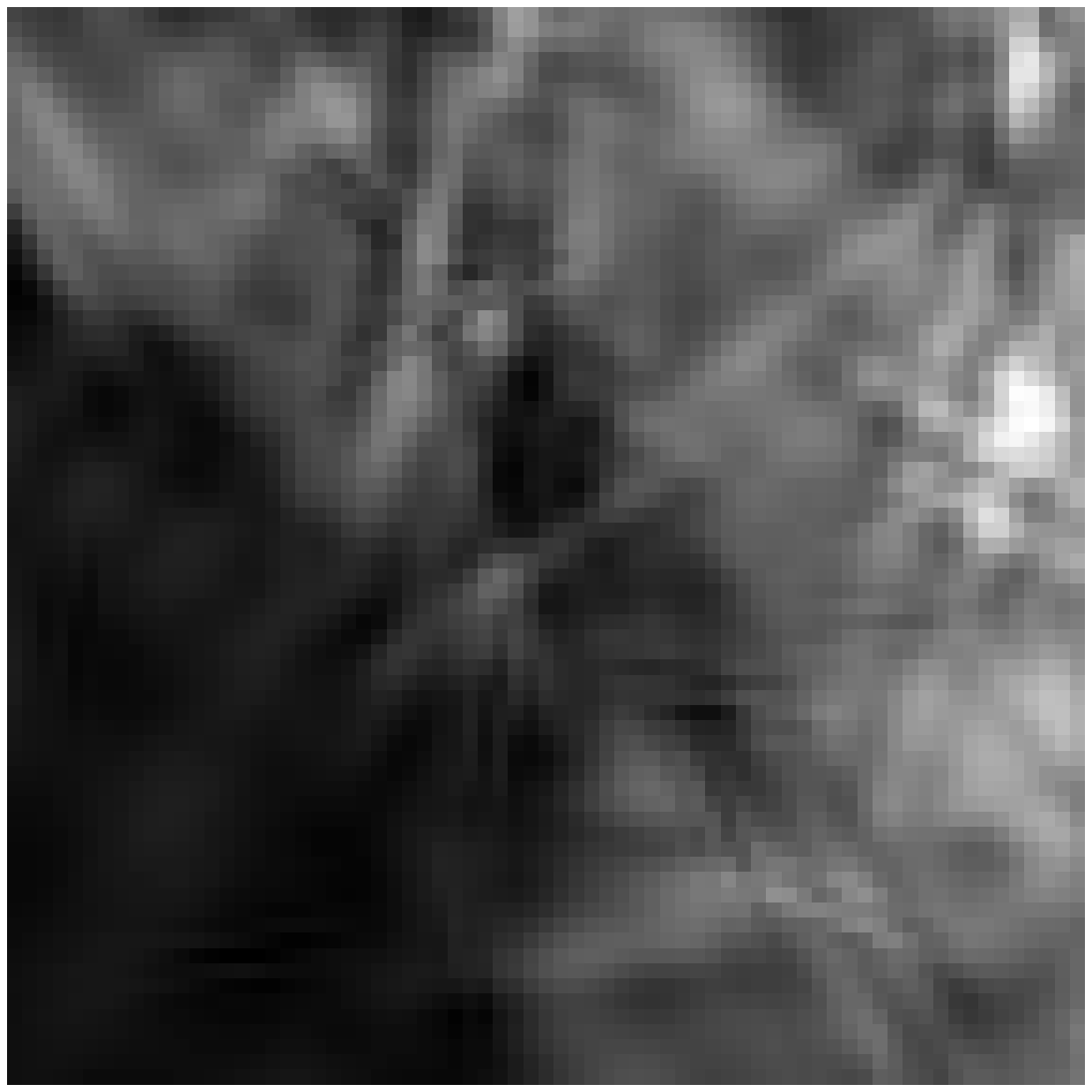}
& \includegraphics[width=4cm]{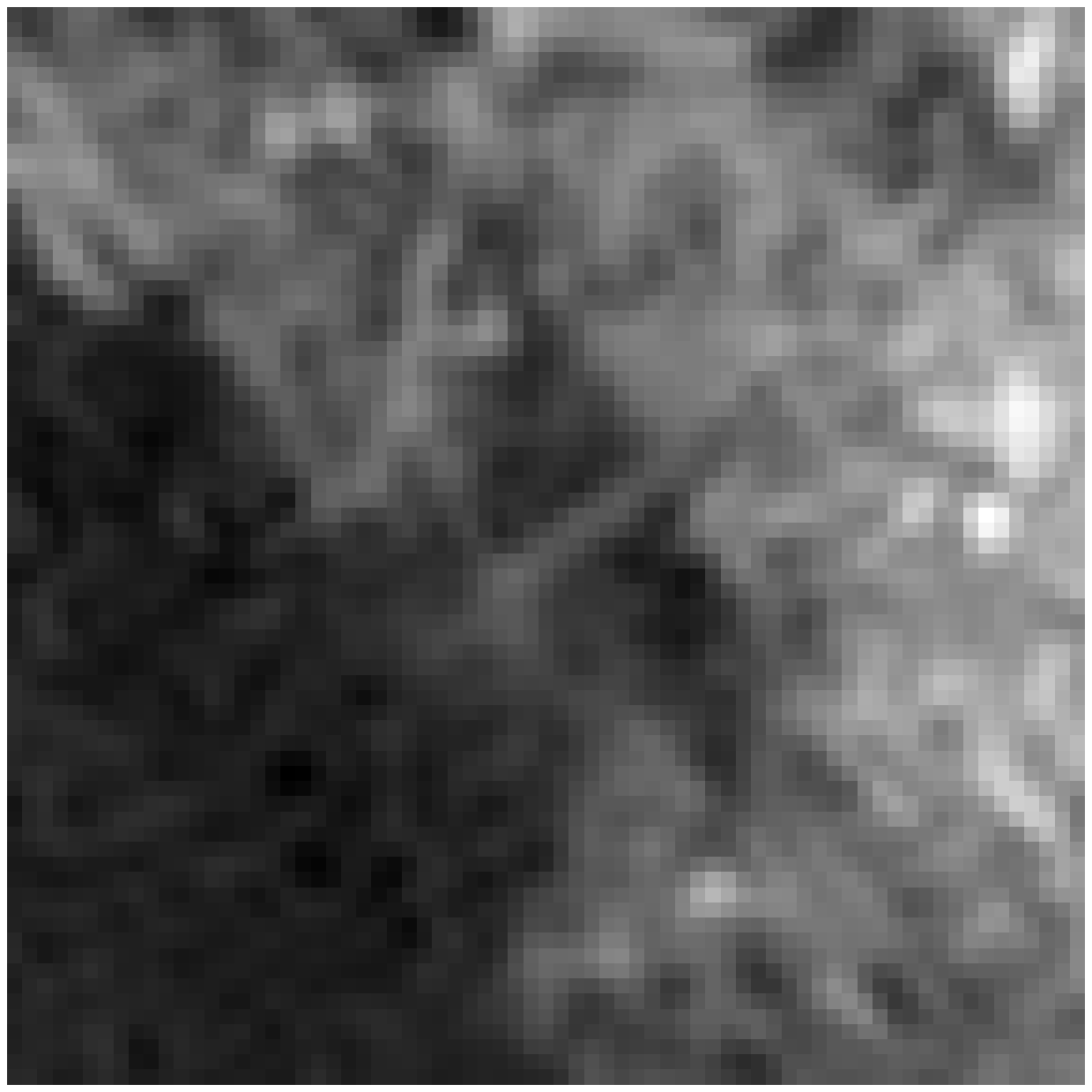}\\
(d) & (e) & (f)\\
 \end{tabular}
\caption{Cropped versions of Tunis image (channel $b=1$, initial SNR equal to $4.66$ dB) and (a) Original image,
(b) Noisy image, (c) Denoised image using ProbShrink red. $(1
\times 1)$, (d) Denoised image using BLS-GSM red. + parent method,
(e) Denoised image using curvelets
and (f) Denoised image using our method (employing a DTT).
\label{fig:zoom1}}
\end{center}
\vspace*{-0.8cm} \end{figure} We can notice that the proposed
method (see Fig. \ref{fig:zoom1}-(f)) allows to better recover
edges whereas the three others (see Fig. \ref{fig:zoom1}-(c,d,e))
result in more blurred images, where some of the original
structures are missing. This is especially visible for the image
denoised with the BLS-GSM estimator (see Fig.
\ref{fig:zoom1}-(d)).

In the following, we focus on the method introduced in this paper and more specifically on the variations of its performance according to the parameter setup.

\subsection{Pre-processing stage}
In order to improve the denoising performance in the multichannel
context, additional linear procedures can be applied.
Actually, different linear pre-processings of the components
may be envisaged:\\
$\bullet$ The simplest idea consists in decorrelating the spectral
components of the image to be estimated in order to process them
separately. Knowing the noise covariance matrix
$\GGamma^{(\mathbf{n})}$, we can deduce the original data
covariance matrix (assumed here to be spatially constant):
$\GGamma^{(\mathbf{s})}=\GGamma^{(\mathbf{r})}\,
-\, \GGamma^{(\mathbf{n})} $,  from the observed data
covariance matrix $\GGamma^{(\mathbf{r})}$. More
precisely, by performing an eigendecomposition of
$\GGamma^{(\mathbf{s})}$, we seek for an orthogonal matrix
$\mathbf{U}^{(\mathbf{s})}$ such that:
$\GGamma^{(\mathbf{s})}=\mathbf{U}^{({\mathbf{s}})}
\mathbf{D}^{(\mathbf{s})} (\mathbf{U}^{(s)})^\top$ where
$\mathbf{D}^{(\mathbf{s})}$ is a diagonal matrix. Then, the
transformed multichannel image is
$((\mathbf{U}^{(\mathbf{s})})^\top\mathbf{r}(\mathbf{k}))_{\mathbf{k}}$
and it is corrupted by a spatially white zero-mean Gaussian noise
with covariance matrix
$(\mathbf{U}^{(\mathbf{s})})^\top\GGamma^{(\mathbf{n})}\mathbf{U}^{(\mathbf{s})}$.
We then proceed to the nonlinear wavelet estimation of the
decorrelated components as described in the previous sections.\\
$\bullet$ Instead of decorrelating the components, we may try to make
them statistically independent or, at least, as independent as
possible. A number of ICA (Independent Component Analysis) methods
have been developed for this purpose in recent years
\cite{Cardoso_J_1993_ieepf_bli_bngs}. In this case, a linear
transform $\mathbf{V}^{(\mathbf{s})}$ (which is not necessarily
orthogonal) is applied to the multichannel data.

The proposed estimator already includes an optimized linear
combination of some of the components of the ROV. It is therefore
expected to provide competitive results w.r.t. techniques
involving some linear pre-processing. In order to make fair
comparisons and evaluate the improvements resulting from the
optimization of the linear part of the estimator, we provide
simulations where the ROV is the same whatever the pre-processing
is (we have chosen the same ROV as in the previous sections). In
addition, when a decorrelation or an ICA is employed, the linear
part of the estimator is chosen equal to the identity. We finally
propose to compare these results with a simple linear MSE
estimator based on a linear combination of coefficients from different channels.

Numerical results displayed in Table \ref{tab:preproc} allow us to
evaluate the proposed approach without optimization of the linear
parameter vector, the same estimator combined with an ICA of the
multichannel data (using the JADE algorithm
\cite{Cardoso_J_1993_ieepf_bli_bngs}) or a pre-decorrelation stage
and, finally our approach with an optimized linear part.
\begin{table}[htbp]
\begin{center}
\caption{Influence of different pre-processings on Tunis image denoising ($\sigma^2=258.9$). Symlets of length $16$ are used. \label{tab:preproc}}
\begin{tabular}{|c ||c || c || c | c | c | c | c |}
\hline
Transform & Channel & $\textrm{SNR}_{\textrm{init}}$ & Without transf. & ICA & Decorrelation & MSE Lin. & Opt. lin. \\
\hline \hline
 & $b=1$ &8.664 & 13.84 & 14.66 & 15.15 & 15.18 & 15.75\\
DWT & $b=2$ & 9.653&14.39 & 15.03 & 15.36 & 15.28 & 15.90\\
& $b=3$ &  8.926& 15.15 & 13.85 & 15.11 & 15.26 & 15.86\\
& Average &9.081 & 14.46 & 14.51 & 15.21 & 15.24 & \textbf{15.84}\\
\hline \hline
 & $b=1$ &8.664  & 14.13 &14.37  & 15.42 & 15.42 & 15.95\\
DTT & $b=2$ & 9.653 & 14.66 &14.67  & 15.64 & 15.53 & 16.10\\
& $b=3$ &  8.926& 15.38 &14.26  & 15.25 & 15.52 & 16.01 \\
& Average &9.081 & 14.72 &14.43  & 15.44 & 15.49 & \textbf{16.02}\\
\hline
\end{tabular}
\end{center}
\vspace*{-0.8cm}
\end{table}
From these results, it is clear
that including some linear processing is useful for multichannel image denoising.
The ICA only brings slight improvements, possibly due to the fact that the associated transform is not orthogonal.
Pre-decorrelating the data significantly increases the SNR,
however the fully optimized version of our estimator remains the most effective method.

\subsection{Influence of the neighborhoods}
The ROV can
be defined as desired and
plays a prominent role in the construction of our estimator.
We study here the influence of different choices of the ROV:
\begin{enumerate}
\item ROV1 corresponds to an inter-component neighborhood. When a
DWT is employed (see Fig. \ref{fig:ROVDWT}(a)), we have
$ \bbf{r}_{j,\mathbf{m}}^{(b)}(\kk) =
[\big(r_{j,\mathbf{m}}^{(b')}(\kk)\big)_{b'}]^\top$,
while for a DTT (see Fig. \ref{fig:ROVDTT}(a)), we use
\begin{align}
\bbf{r}_{j,\mathbf{m}}^{(b)}(\kk) &=
[\big(r_{j,\mathbf{m}}^{(b')}(\kk)\big)_{b'}\,,
\big(r_{j,\mathbf{m}}^{\HH(b')}(\kk)\big)_{b'} ]^\top \quad
\mbox{and}
\quad \bbf{u}_{j,\mathbf{m}}^{(b)}(\kk) = [\big(u_{j,\mathbf{m}}^{(b')}(\kk)\big)_{b'}]^\top\\
\bbf{r}_{j,\mathbf{m}}^{\HH(b)}(\kk) &= [
\big(r_{j,\mathbf{m}}^{\HH(b')}(\kk)\big)_{b'}\,,
\big(r_{j,\mathbf{m}}^{(b')}(\kk)\big)_{b'} ]^\top \quad \quad
\quad \;\; \bbf{u}_{j,\mathbf{m}}^{\HH(b)}(\kk) =
[\big(u_{j,\mathbf{m}}^{\HH(b')}(\kk)\big)_{b'}]^\top.
\end{align}
\item ROV2 corresponds to a combination
of a spatial $3\times3$ and an inter-component neighborhood as
considered in the previous sections and shown in Figs.
\ref{fig:ROVDWT}(b) and \ref{fig:ROVDTT}(b).
\end{enumerate}

\begin{figure}[!h]
\begin{center}
\begin{tabular}{cc}
\input{ROV_spec.pstex_t} & \input{ROV_spat+spec.pstex_t} \\
(a) & (b) \\
\end{tabular}
\end{center}
\caption{Representation of the different considered ROVs in the
DWT domain (the black triangle will be estimated taking into
account the white ones); (a) ROV1 the purely inter-component one
and (b) ROV2 combining inter-component
and spatial dependencies. \label{fig:ROVDWT}} \vspace*{-0.2cm}
\end{figure}
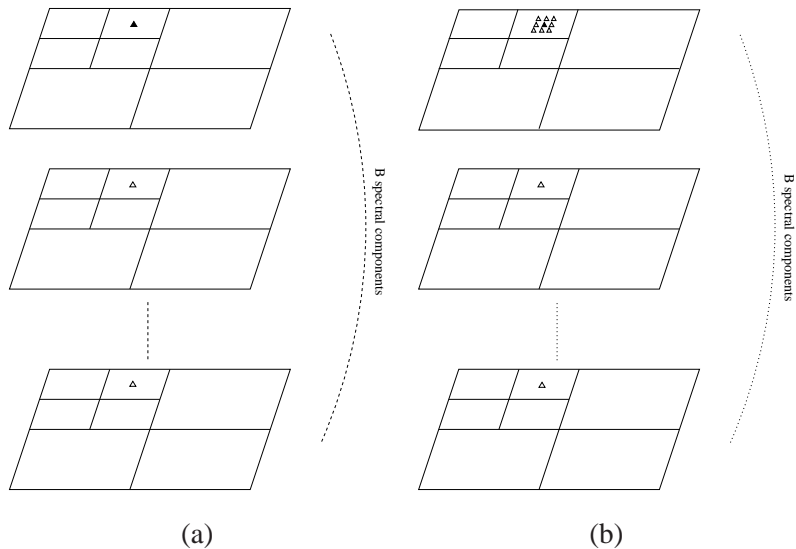
\begin{figure}[!h]
\begin{center}
\begin{tabular}{cc}
\input{ROVDTT_spec.pstex_t} &  \input{ROVDTT_spat+spec.pstex_t} \\
(a) & (b) \\
\end{tabular}
\end{center}
\caption{Representation of the different considered ROVs in the
DTT domain, with and without post-processing stage (the black
triangle will be estimated taking into account the white ones);
(a) ROV1 the purely inter-component one and (b) ROV2 combining inter-component and spatial dependencies.
\label{fig:ROVDTT}} \end{figure}
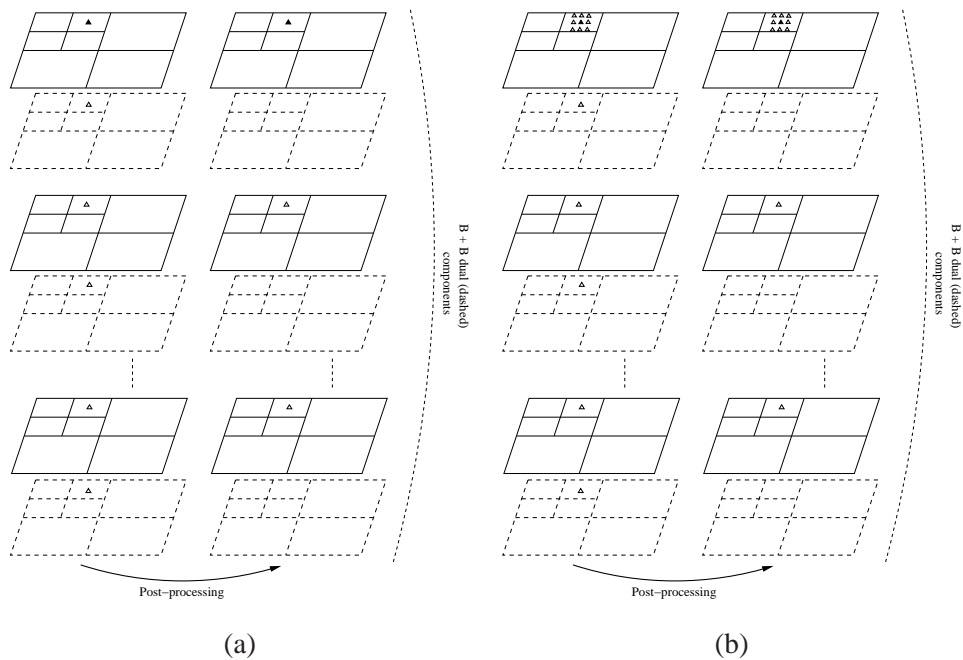

The linear part of the estimator is defined as in Section \ref{se:compare}.\\
The corresponding results are given in Table \ref{tab:neigh}.
\begin{table}[htbp]
\begin{center}
\caption{Influence of the neighborhood in Tunis image denoising
(average values computed over $3$ channels are provided and $\sigma^2=258.9$) using
symlets (length $16$) (top) and AC filter bank (length $16$)
(bottom). \label{tab:neigh}}
\begin{tabular}{|c || c || c | c || c | c || c | c |}
\hline
Transform & $\textrm{SNR}_{\textrm{init}}$  & ROV1  & ROV2& Transform & $\textrm{SNR}_{\textrm{init}}$  & ROV1  & ROV2\\
\hline \hline 
DWT (symlets)&9.081  & 15.42 & \textbf{15.84} & DWT (AC)&9.081 & 15.48&  \textbf{15.77} \\
DTT (symlets)&9.081&  15.77   & \textbf{16.02} & DTT (AC)&9.081 & 15.88 & \textbf{16.02}\\
\hline
\end{tabular}
\end{center}
\vspace*{-0.8cm}
\end{table}

In order to compare different possible wavelet choices, the
results are provided both for  symlets of length $16$ and a
$4$-band filter bank given in
\cite{Alkin_O_1995_tsp_des_embclpprp} which is denoted by AC.
These results can also be compared with the ones given in Section
\ref{se:compare} where Daubechies filters of length
$8$ are used.\\
Concerning the neighborhood influence, we note that taking into
account spatial dependence leads to a significant improvement
w.r.t. inter-component dependence. In all cases, combining
spectral and spatial neighborhood leads to the best results.\\
Concerning the wavelet choice, it appears that the 4-band AC
wavelets yield slightly better results than the dyadic symlets
choosing ROV1 and equivalent results choosing ROV3.
Both outperform Daubechies wavelets wathever the ROV chosen.

\subsection{Various noise levels}
In this section, we consider that the image channels are corrupted at different
noise levels.
Thus, the noise is spatially white, zero-mean, Gaussian with covariance matrix
$\GGamma^{(\mathbf{n})}_2= \text{Diag}(\sigma_1^2, \ldots,\sigma_B^2)$.
\begin{table}[htbp]
\begin{center}
\caption{Denoising results on Tunis image considering
$\boldsymbol{\Gamma}^{(n)}_2$ and using symlets (length $16$). \label{tab:sigdif}}
\begin{tabular}{|c ||c | c || c | c || c | c |}
\hline
Channel & $\sigma_b^2$ & $\textrm{SNR}_{\textrm{init}}$ & Surevect & Proposed DWT & Surevect DTT  & Proposed DTT\\
\hline \hline
$b=1$ & 25.89 & 18.66  & 20.58  & 21.15 & 20.84 & 21.24\\
$b=2$ & 258.9 & 9.653 & 18.53 & 18.63 & 18.76 & 18.84\\
$b=3$ & 491.9 & 6.138 & 14.20 &  14.57 & 14.55 & 14.71\\
Average & & 11.49 & 17.76 & 18.12 & 18.05 & \textbf{18.26}\\
\hline
\end{tabular}
\end{center}
\vspace*{-0.8cm}
\end{table}

The resulting numerical results are displayed in Table
\ref{tab:sigdif} with the corresponding noise levels, when our
estimator is used with ROV2. Noticeable differences can be
observed when comparing Surevect with our method both considering
DWT and DTT transforms.

\subsection{Increased number of channels}
A strong advantage of the proposed method is that, unlike many
multicomponent approaches limited to RGB ($3$ components) images,
it may process any kind of multichannel images, whatever the
number of channels is. We consider here the $6$ channel Trento
image. We apply the Surevect estimator (both using DWT and DTT), 
the BLS-GSM estimator (taking into account the parent coefficient),
and our estimator using ROV2.
\begin{table}[htbp]
\begin{center}
\caption{Results obtained applying different estimators on Trento
image ($\sigma^2=258.9$). \label{tab:band}}
\begin{tabular}{|c || c || c | c || c | c | c |}
\hline
Channel & $\textrm{SNR}_{\textrm{init}}$  & Surevect & Proposed & BLS-GSM red & Surevect  & Proposed \\
& & & DWT & + parent & DTT & DTT \\
\hline \hline
$b=1$ & -2.907  &8.661  &8.945  & 8.311 & 8.984 & 9.255  \\
$b=2$ & -6.878  &8.375  &8.430  & 6.536 & 8.805 & 8.876 \\
$b=3$ & -3.836  &8.288  &8.430  & 7.341 & 8.647 & 8.749  \\
$b=4$ & 2.428   &9.525  &9.796  & 9.836 & 9.901 & 10.00\\
$b=5$ & 4.765   &11.18  &11.53  & 11.38 & 11.61 & 11.78 \\
$b=6$ & -1.560  &9.545 &9.700  & 8.167 & 9.945 & 10.02 \\
Average & -1.331 &9.262  &9.472  & 8.596 & 9.649  & \textbf{9.780} \\
\hline
\end{tabular}
\end{center}
\vspace*{-0.8cm}
\end{table}
From the results provided in Table \ref{tab:band}, we see
that, while the number of channels is increased, our method still
outperforms the other ones especially when a DTT is used. With the
increase of the number of channels, the reduced redundancy of the
DTT becomes another attractive feature of the

\noindent proposed approach.

\vspace*{-0.4cm}

\section{Conclusion}
\label{sec:concl}

In this paper, we have proposed a nonlinear Stein based estimator
for wavelet denoising of multichannel data. Due to its flexible
form, the considered estimator generalizes many existing methods,
in particular block-based ones. Although the proposed approach has
been applied to satellite images, it could also be used in any
multivariate signal denoising problem. Besides, the estimator has
been used in conjunction with real dual-tree wavelet transforms
but complex ones or other frame decompositions could be envisaged
as well. In the context of frame representations, it should
however be noticed that the proposed estimator minimizes the risk
over the frame coefficients and not on the reconstructed signal,
which may be suboptimal \cite{Luisier_F_2007_tip_new_sureaidisowt,Raphan_M_2007_icip_opt_drb}.
Another
question that should be investigated in a future work is the ability
of the proposed framework to exploit inter-scale dependencies in
addition to spatial and inter-component ones, as considered in
\cite{Luisier_F_2007_tip_new_sureaidisowt} for the mono-channel
case. In order to obtain an interscale denoising method,
an appropriate ROV should be defined and the interscale statistics
of the noise should be available.

{\scriptsize{
}

\end{document}

%% file: model_wavelets.pstex_t
\begin{picture}(0,0)%
\includegraphics{model_wavelets.pstex}%
\end{picture}%
\setlength{\unitlength}{2486sp}%
\begingroup\makeatletter\ifx\SetFigFont\undefined%
\gdef\SetFigFont#1#2#3#4#5{%
  \reset@font\fontsize{#1}{#2pt}%
  \fontfamily{#3}\fontseries{#4}\fontshape{#5}%
  \selectfont}%
\fi\endgroup%
\begin{picture}(5016,3696)(46,-2914)
\put(3061,614){\makebox(0,0)[lb]{\smash{{\SetFigFont{8}{9.6}{\rmdefault}{\mddefault}{\updefault}$r_{j,\mathbf{m}}^{(b)}$}}}}
\put(1666,-2086){\makebox(0,0)[lb]{\smash{{\SetFigFont{7}{8.4}{\rmdefault}{\mddefault}{\updefault}$\mathcal{WT}$}}}}
\put(2926,-16){\makebox(0,0)[lb]{\smash{{\SetFigFont{7}{8.4}{\rmdefault}{\mddefault}{\updefault}$\mathcal{WT}$}}}}
\put(1891,-2851){\makebox(0,0)[lb]{\smash{{\SetFigFont{8}{9.6}{\rmdefault}{\mddefault}{\updefault}$n_{j,\mathbf{m}}^{(b)}$}}}}
\put(541,-16){\makebox(0,0)[lb]{\smash{{\SetFigFont{7}{8.4}{\rmdefault}{\mddefault}{\updefault}$\mathcal{WT}$}}}}
\put(721,614){\makebox(0,0)[lb]{\smash{{\SetFigFont{8}{9.6}{\rmdefault}{\mddefault}{\updefault}$s_{j,\mathbf{m}}^{(b)}$}}}}
\put(4186,-556){\makebox(0,0)[lb]{\smash{{\SetFigFont{8}{9.6}{\rmdefault}{\mddefault}{\updefault}$r^{(b)}$}}}}
\put( 46,-1546){\makebox(0,0)[lb]{\smash{{\SetFigFont{7}{8.4}{\rmdefault}{\mddefault}{\updefault}{\color[rgb]{0,0,0}iid ${\mathcal{N}}({\mathbf{0}},{\mathbf{\Gamma}}^{(n)})$}%
}}}}
\put( 46,-1231){\makebox(0,0)[lb]{\smash{{\SetFigFont{8}{9.6}{\rmdefault}{\mddefault}{\updefault}$n^{(b)}$}}}}
\put( 46,-556){\makebox(0,0)[lb]{\smash{{\SetFigFont{8}{9.6}{\rmdefault}{\mddefault}{\updefault}$s^{(b)}$}}}}
\end{picture}%

%% file: model_dualtree.pstex_t
\begin{picture}(0,0)%
\includegraphics{model_dualtree.pstex}%
\end{picture}%
\setlength{\unitlength}{2486sp}%
\begingroup\makeatletter\ifx\SetFigFont\undefined%
\gdef\SetFigFont#1#2#3#4#5{%
  \reset@font\fontsize{#1}{#2pt}%
  \fontfamily{#3}\fontseries{#4}\fontshape{#5}%
  \selectfont}%
\fi\endgroup%
\begin{picture}(4971,3696)(1,-2914)
\put(136,614){\makebox(0,0)[lb]{\smash{{\SetFigFont{8}{9.6}{\rmdefault}{\mddefault}{\updefault}$(s_{j,\mathbf{m}}^{(b)},s_{j,\mathbf{m}}^{\HH(b)})$}}}}
\put(2521,614){\makebox(0,0)[lb]{\smash{{\SetFigFont{8}{9.6}{\rmdefault}{\mddefault}{\updefault}$(r_{j,\mathbf{m}}^{(b)},r_{j,\mathbf{m}}^{\HH(b)})$}}}}
\put(2836,-16){\makebox(0,0)[lb]{\smash{{\SetFigFont{7}{8.4}{\rmdefault}{\mddefault}{\updefault}$\mathcal{DTT}$}}}}
\put(496,-16){\makebox(0,0)[lb]{\smash{{\SetFigFont{7}{8.4}{\rmdefault}{\mddefault}{\updefault}$\mathcal{DTT}$}}}}
\put(1621,-2086){\makebox(0,0)[lb]{\smash{{\SetFigFont{7}{8.4}{\rmdefault}{\mddefault}{\updefault}$\mathcal{DTT}$}}}}
\put(1351,-2851){\makebox(0,0)[lb]{\smash{{\SetFigFont{8}{9.6}{\rmdefault}{\mddefault}{\updefault}$(n_{j,\mathbf{m}}^{(b)},n_{j,\mathbf{m}}^{\HH(b)})$}}}}
\put(4096,-556){\makebox(0,0)[lb]{\smash{{\SetFigFont{8}{9.6}{\rmdefault}{\mddefault}{\updefault}$r^{(b)}$}}}}
\put(  1,-1501){\makebox(0,0)[lb]{\smash{{\SetFigFont{7}{8.4}{\rmdefault}{\mddefault}{\updefault}{\color[rgb]{0,0,0}iid ${\mathcal{N}}({\mathbf{0}},{\mathbf{\Gamma}}^{(n)})$}%
}}}}
\put(  1,-1231){\makebox(0,0)[lb]{\smash{{\SetFigFont{8}{9.6}{\rmdefault}{\mddefault}{\updefault}$n^{(b)}$}}}}
\put(  1,-556){\makebox(0,0)[lb]{\smash{{\SetFigFont{8}{9.6}{\rmdefault}{\mddefault}{\updefault}$s^{(b)}$}}}}
\end{picture}%

%% file: fb_M_ortho_dual.pstex_t
\begin{picture}(0,0)%
\includegraphics{fb_M_ortho_dual.pstex}%
\end{picture}%
\setlength{\unitlength}{2486sp}%
\begingroup\makeatletter\ifx\SetFigFont\undefined%
\gdef\SetFigFont#1#2#3#4#5{%
  \reset@font\fontsize{#1}{#2pt}%
  \fontfamily{#3}\fontseries{#4}\fontshape{#5}%
  \selectfont}%
\fi\endgroup%
\begin{picture}(5084,3419)(609,-8858)
\put(2026,-7441){\makebox(0,0)[b]{\smash{{\SetFigFont{7}{8.4}{\familydefault}{\mddefault}{\updefault}$G^*_0$}}}}
\put(2026,-7981){\makebox(0,0)[b]{\smash{{\SetFigFont{7}{8.4}{\familydefault}{\mddefault}{\updefault}$G^*_1$}}}}
\put(2026,-8701){\makebox(0,0)[b]{\smash{{\SetFigFont{7}{8.4}{\familydefault}{\mddefault}{\updefault}$G^*_{M-1}$}}}}
\put(2836,-7441){\makebox(0,0)[b]{\smash{{\SetFigFont{7}{8.4}{\familydefault}{\mddefault}{\updefault}$\downarrow\!M$}}}}
\put(2836,-7981){\makebox(0,0)[b]{\smash{{\SetFigFont{7}{8.4}{\familydefault}{\mddefault}{\updefault}$\downarrow\!M$}}}}
\put(2836,-8701){\makebox(0,0)[b]{\smash{{\SetFigFont{7}{8.4}{\familydefault}{\mddefault}{\updefault}$\downarrow\!M$}}}}
\put(4006,-7441){\makebox(0,0)[b]{\smash{{\SetFigFont{7}{8.4}{\familydefault}{\mddefault}{\updefault}$\uparrow\!M$}}}}
\put(4006,-7981){\makebox(0,0)[b]{\smash{{\SetFigFont{7}{8.4}{\familydefault}{\mddefault}{\updefault}$\uparrow\!M$}}}}
\put(4006,-8701){\makebox(0,0)[b]{\smash{{\SetFigFont{7}{8.4}{\familydefault}{\mddefault}{\updefault}$\uparrow\!M$}}}}
\put(4816,-7441){\makebox(0,0)[b]{\smash{{\SetFigFont{7}{8.4}{\familydefault}{\mddefault}{\updefault}$G_0$}}}}
\put(4816,-7981){\makebox(0,0)[b]{\smash{{\SetFigFont{7}{8.4}{\familydefault}{\mddefault}{\updefault}$G_1$}}}}
\put(4816,-8701){\makebox(0,0)[b]{\smash{{\SetFigFont{7}{8.4}{\familydefault}{\mddefault}{\updefault}$G_{M-1}$}}}}
\put(1801,-5686){\makebox(0,0)[b]{\smash{{\SetFigFont{7}{8.4}{\familydefault}{\mddefault}{\updefault}$H^*_0$}}}}
\put(1801,-6226){\makebox(0,0)[b]{\smash{{\SetFigFont{7}{8.4}{\familydefault}{\mddefault}{\updefault}$H^*_1$}}}}
\put(1801,-6946){\makebox(0,0)[b]{\smash{{\SetFigFont{7}{8.4}{\familydefault}{\mddefault}{\updefault}$H^*_{M-1}$}}}}
\put(2611,-5686){\makebox(0,0)[b]{\smash{{\SetFigFont{7}{8.4}{\familydefault}{\mddefault}{\updefault}$\downarrow\!M$}}}}
\put(2611,-6226){\makebox(0,0)[b]{\smash{{\SetFigFont{7}{8.4}{\familydefault}{\mddefault}{\updefault}$\downarrow\!M$}}}}
\put(3781,-5686){\makebox(0,0)[b]{\smash{{\SetFigFont{7}{8.4}{\familydefault}{\mddefault}{\updefault}$\uparrow\!M$}}}}
\put(3781,-6226){\makebox(0,0)[b]{\smash{{\SetFigFont{7}{8.4}{\familydefault}{\mddefault}{\updefault}$\uparrow\!M$}}}}
\put(3781,-6946){\makebox(0,0)[b]{\smash{{\SetFigFont{7}{8.4}{\familydefault}{\mddefault}{\updefault}$\uparrow\!M$}}}}
\put(4591,-5686){\makebox(0,0)[b]{\smash{{\SetFigFont{7}{8.4}{\familydefault}{\mddefault}{\updefault}$H_0$}}}}
\put(4591,-6226){\makebox(0,0)[b]{\smash{{\SetFigFont{7}{8.4}{\familydefault}{\mddefault}{\updefault}$H_1$}}}}
\put(4591,-6946){\makebox(0,0)[b]{\smash{{\SetFigFont{7}{8.4}{\familydefault}{\mddefault}{\updefault}$H_{M-1}$}}}}
\put(2611,-6946){\makebox(0,0)[b]{\smash{{\SetFigFont{7}{8.4}{\familydefault}{\mddefault}{\updefault}$\downarrow\!M$}}}}
\end{picture}%

%% file: Mbanddualtree2D.pstex_t
\begin{picture}(0,0)%
\includegraphics{Mbanddualtree2D.pstex}%
\end{picture}%
\setlength{\unitlength}{2368sp}%
\begingroup\makeatletter\ifx\SetFigFont\undefined%
\gdef\SetFigFont#1#2#3#4#5{%
  \reset@font\fontsize{#1}{#2pt}%
  \fontfamily{#3}\fontseries{#4}\fontshape{#5}%
  \selectfont}%
\fi\endgroup%
\begin{picture}(10684,5112)(751,-7561)
\put(10051,-3136){\makebox(0,0)[lb]{\smash{{\SetFigFont{7}{8.4}{\familydefault}{\mddefault}{\updefault}$u_{2,m_1,m_2}$}}}}
\put(4050,-3436){\makebox(0,0)[b]{\smash{{\SetFigFont{7}{8.4}{\familydefault}{\mddefault}{\updefault}filter bank}}}}
\put(2219,-3444){\makebox(0,0)[b]{\smash{{\SetFigFont{7}{8.4}{\familydefault}{\mddefault}{\updefault}$(F_1)$}}}}
\put(2221,-3259){\makebox(0,0)[b]{\smash{{\SetFigFont{7}{8.4}{\familydefault}{\mddefault}{\updefault}Prefiltering}}}}
\put(2236,-4984){\makebox(0,0)[b]{\smash{{\SetFigFont{7}{8.4}{\familydefault}{\mddefault}{\updefault}Prefiltering}}}}
\put(2234,-5169){\makebox(0,0)[b]{\smash{{\SetFigFont{7}{8.4}{\familydefault}{\mddefault}{\updefault}$(F_2)$}}}}
\put(8896,-3511){\makebox(0,0)[b]{\smash{{\SetFigFont{7}{8.4}{\familydefault}{\mddefault}{\updefault}Linear}}}}
\put(8880,-3811){\makebox(0,0)[b]{\smash{{\SetFigFont{7}{8.4}{\familydefault}{\mddefault}{\updefault}combinaison}}}}
\put(8880,-4141){\makebox(0,0)[b]{\smash{{\SetFigFont{7}{8.4}{\familydefault}{\mddefault}{\updefault}of the}}}}
\put(8881,-4426){\makebox(0,0)[b]{\smash{{\SetFigFont{7}{8.4}{\familydefault}{\mddefault}{\updefault}subbands}}}}
\put(6758,-2911){\makebox(0,0)[b]{\smash{{\SetFigFont{7}{8.4}{\familydefault}{\mddefault}{\updefault}$M$-band}}}}
\put(6758,-3211){\makebox(0,0)[b]{\smash{{\SetFigFont{7}{8.4}{\familydefault}{\mddefault}{\updefault}filter bank}}}}
\put(8896,-5521){\makebox(0,0)[b]{\smash{{\SetFigFont{7}{8.4}{\familydefault}{\mddefault}{\updefault}Linear}}}}
\put(8880,-5821){\makebox(0,0)[b]{\smash{{\SetFigFont{7}{8.4}{\familydefault}{\mddefault}{\updefault}combinaison}}}}
\put(8880,-6151){\makebox(0,0)[b]{\smash{{\SetFigFont{7}{8.4}{\familydefault}{\mddefault}{\updefault}of the}}}}
\put(8881,-6436){\makebox(0,0)[b]{\smash{{\SetFigFont{7}{8.4}{\familydefault}{\mddefault}{\updefault}subbands}}}}
\put(3076,-3136){\makebox(0,0)[b]{\smash{{\SetFigFont{7}{8.4}{\familydefault}{\mddefault}{\updefault}$r_{0,0,0}$}}}}
\put(5101,-2911){\makebox(0,0)[b]{\smash{{\SetFigFont{7}{8.4}{\familydefault}{\mddefault}{\updefault}$r_{1,0,0}$}}}}
\put(5101,-4636){\makebox(0,0)[b]{\smash{{\SetFigFont{7}{8.4}{\familydefault}{\mddefault}{\updefault}$r^{\mathrm{H}}_{1,0,0}$}}}}
\put(3076,-4861){\makebox(0,0)[b]{\smash{{\SetFigFont{7}{8.4}{\familydefault}{\mddefault}{\updefault}$r^{\mathrm{H}}_{0,0,0}$}}}}
\put(6451,-4486){\makebox(0,0)[lb]{\smash{{\SetFigFont{7}{8.4}{\rmdefault}{\mddefault}{\updefault}"Dual"}}}}
\put(6751,-4711){\makebox(0,0)[b]{\smash{{\SetFigFont{7}{8.4}{\familydefault}{\mddefault}{\updefault}$M$-band}}}}
\put(6751,-4936){\makebox(0,0)[b]{\smash{{\SetFigFont{7}{8.4}{\familydefault}{\mddefault}{\updefault}filter bank}}}}
\put(4051,-5236){\makebox(0,0)[b]{\smash{{\SetFigFont{7}{8.4}{\familydefault}{\mddefault}{\updefault}filter bank}}}}
\put(3751,-4786){\makebox(0,0)[lb]{\smash{{\SetFigFont{7}{8.4}{\rmdefault}{\mddefault}{\updefault}"Dual"}}}}
\put(4051,-5011){\makebox(0,0)[b]{\smash{{\SetFigFont{7}{8.4}{\familydefault}{\mddefault}{\updefault}$M$-band}}}}
\put(10051,-6136){\makebox(0,0)[lb]{\smash{{\SetFigFont{7}{8.4}{\familydefault}{\mddefault}{\updefault}$(u^{\mathrm{H}}_{1,m_1,m_2})_{(m_1,m_2)\neq(0,0)}$}}}}
\put(10051,-5386){\makebox(0,0)[lb]{\smash{{\SetFigFont{7}{8.4}{\familydefault}{\mddefault}{\updefault}$(u_{1,m_1,m_2})_{(m_1,m_2)\neq(0,0)}$}}}}
\put(10051,-4411){\makebox(0,0)[lb]{\smash{{\SetFigFont{7}{8.4}{\familydefault}{\mddefault}{\updefault}$u^{\mathrm{H}}_{2,m_1,m_2}$}}}}
\put(4050,-3136){\makebox(0,0)[b]{\smash{{\SetFigFont{7}{8.4}{\familydefault}{\mddefault}{\updefault}$M$-band}}}}
\end{picture}%

%% file: ROV_spec.pstex_t
\begin{picture}(0,0)%
\includegraphics{ROV_spec.pstex}%
\end{picture}%
\setlength{\unitlength}{1658sp}%
\begingroup\makeatletter\ifx\SetFigFont\undefined%
\gdef\SetFigFont#1#2#3#4#5{%
  \reset@font\fontsize{#1}{#2pt}%
  \fontfamily{#3}\fontseries{#4}\fontshape{#5}%
  \selectfont}%
\fi\endgroup%
\begin{picture}(5643,7224)(2389,-6973)
\end{picture}%

%% file: ROV_spat+spec.pstex_t
\begin{picture}(0,0)%
\includegraphics{ROV_spat+spec.pstex}%
\end{picture}%
\setlength{\unitlength}{1658sp}%
\begingroup\makeatletter\ifx\SetFigFont\undefined%
\gdef\SetFigFont#1#2#3#4#5{%
  \reset@font\fontsize{#1}{#2pt}%
  \fontfamily{#3}\fontseries{#4}\fontshape{#5}%
  \selectfont}%
\fi\endgroup%
\begin{picture}(5643,7209)(2389,-6973)
\end{picture}%

%% file: ROVDTT_spec.pstex_t
\begin{picture}(0,0)%
\includegraphics{ROVDTT_spec.pstex}%
\end{picture}%
\setlength{\unitlength}{1776sp}%
\begingroup\makeatletter\ifx\SetFigFont\undefined%
\gdef\SetFigFont#1#2#3#4#5{%
  \reset@font\fontsize{#1}{#2pt}%
  \fontfamily{#3}\fontseries{#4}\fontshape{#5}%
  \selectfont}%
\fi\endgroup%
\begin{picture}(6426,8252)(124,-7480)
\end{picture}%

%% file: ROVDTT_spat+spec.pstex_t
\begin{picture}(0,0)%
\includegraphics{ROVDTT_spat+spec.pstex}%
\end{picture}%
\setlength{\unitlength}{1776sp}%
\begingroup\makeatletter\ifx\SetFigFont\undefined%
\gdef\SetFigFont#1#2#3#4#5{%
  \reset@font\fontsize{#1}{#2pt}%
  \fontfamily{#3}\fontseries{#4}\fontshape{#5}%
  \selectfont}%
\fi\endgroup%
\begin{picture}(6426,8252)(124,-7480)
\end{picture}%